\documentclass[11pt]{article}
\usepackage[letterpaper, margin=1in]{geometry}
\usepackage[utf8]{inputenc}
\usepackage{newtxtext}
\def\PrintMode{0}

\usepackage{tikz}
\usetikzlibrary{positioning}
\usepackage{pgffor}
\usetikzlibrary{decorations.pathreplacing}
\usepackage{expl3}
\usepackage[T1]{fontenc}
\usepackage{graphicx}
\usepackage{float}
\usepackage[lined,ruled,linesnumbered]{algorithm2e}
\usepackage{xspace}
\usepackage[utf8]{inputenc}
\usepackage[noadjust]{cite}
\usepackage{amsmath, amssymb, amsfonts, amsthm}
\usepackage{enumerate}
\usepackage{xcolor}
\usepackage[all]{xy}
\usepackage[linktocpage=true,pagebackref=true,colorlinks,linkcolor=magenta,citecolor=blue,bookmarks,bookmarksopen,bookmarksnumbered]{hyperref}
\usepackage{bm}
\usepackage{bbm}
\usepackage[normalem]{ulem}
\usepackage{paralist}
\usepackage{gensymb}
\usepackage{thmtools}
\usepackage{rotating}
\usepackage{mdframed}
\usepackage{multirow}
\usepackage{appendix}
\usepackage{tabularx}
\usepackage{verbatim}
\usepackage{mathrsfs}
\usepackage{indentfirst}
\usepackage{mleftright}
\usepackage[nottoc]{tocbibind}
\usepackage{complexity}
\usepackage{tablefootnote}
\usepackage{epigraph}
\usepackage{autobreak}

\ifnum\PrintMode=1

\usepackage{savetrees}
\else
\fi

\graphicspath{{./figs/}}

\usepackage[capitalize]{cleveref}

\newtheorem{theorem}{Theorem}[section]
\newtheorem{lemma}[theorem]{Lemma}
\newtheorem{proposition}[theorem]{Proposition}
\newtheorem{corollary}[theorem]{Corollary}
\newtheorem{claim}[theorem]{Claim}

\newtheorem{fact}[theorem]{Fact}

\newtheorem{open}[theorem]{Open Problem}
\newtheorem{conjecture}[theorem]{Conjecture}

\theoremstyle{definition}
\newtheorem{definition}[theorem]{Definition}

\newtheorem{remark}[theorem]{Remark}
\newtheorem*{remark*}{Remark}

\renewcommand*\backref[1]{\ifx#1\relax \else (cit.~on p.~#1) \fi} 

\makeatletter
\def\moverlay{\mathpalette\mov@rlay}
\def\mov@rlay#1#2{\leavevmode\vtop{%
		\baselineskip\z@skip \lineskiplimit-\maxdimen
		\ialign{\hfil$\m@th#1##$\hfil\cr#2\crcr}}}
\newcommand{\charfusion}[3][\mathord]{
	#1{\ifx#1\mathop\vphantom{#2}\fi
		\mathpalette\mov@rlay{#2\cr#3}
	}
	\ifx#1\mathop\expandafter\displaylimits\fi}
\makeatother

\newcommand{\bigcupdot}{\charfusion[\mathop]{\bigcup}{\cdot}}

\renewcommand{\poly}{\mathrm{poly}}

\renewcommand{\emptyset}{\varnothing}

\newlang{\MCSP}{MCSP}
\newlang{\MFSP}{MFSP}
\newlang{\MKtP}{MKtP}
\newlang{\MKTP}{MKTP}
\newlang{\itrMCSP}{itrMCSP}
\newlang{\itrMKTP}{itrMKTP}
\newlang{\itrMINKT}{itrMINKT}
\newlang{\MINKT}{MINKT}
\newlang{\MINK}{MINK}
\newlang{\MINcKT}{MINcKT}
\newlang{\CMD}{CMD}
\newlang{\DCMD}{DCMD}
\newlang{\CGL}{CGL}
\newlang{\PARITY}{PARITY}
\newlang{\Empty}{\textsc{Empty}}
\newlang{\Avoid}{\textsc{Avoid}}
\newlang{\Sparsification}{\textsc{Sparsification}}
\newlang{\HamEst}{\mathsf{HammingEst}}
\newlang{\HamHit}{\mathsf{HammingHit}}
\newlang{\CktEval}{\textsc{Circuit-Eval}}
\newlang{\Hard}{\textsc{Hard}}
\newlang{\cHard}{\textsc{cHard}}
\newlang{\CAPP}{CAPP}
\newlang{\GapUNSAT}{GapUNSAT}
\newlang{\OV}{OV}
\renewlang{\PCP}{PCP}
\newlang{\PCPP}{PCPP}
\newclass{\Avg}{Avg}
\newclass{\ZPEXP}{ZPEXP}
\newclass{\DLOGTIME}{DLOGTIME}
\newclass{\ALOGTIME}{ALOGTIME}
\newclass{\ATIME}{ATIME}
\newclass{\SZKA}{SZKA}
\newclass{\Laconic}{Laconic\text{-}}
\newclass{\APEPP}{APEPP}
\newclass{\SAPEPP}{SAPEPP}
\newclass{\TFSigma}{TF\Sigma}
\newclass{\NTIMEGUESS}{NTIMEGUESS}

\newlang{\Formula}{Formula}
\newlang{\THR}{THR}
\newlang{\EMAJ}{EMAJ}
\newlang{\MAJ}{MAJ}
\newlang{\SYM}{SYM}
\newlang{\DOR}{DOR}
\newlang{\ETHR}{ETHR}
\newlang{\Midbit}{Midbit}
\newlang{\LCS}{LCS}
\newlang{\TAUT}{TAUT}
\newlang{\Poly}{\text{-}Poly}

\newcommand{\F}{\mathbb{F}}

\DeclareMathOperator*{\Span}{\mathrm{Span}}

\renewcommand{\epsilon}{\varepsilon}
\newcommand{\eps}{\epsilon}

\usepackage[draft, inline,marginclue,index]{fixme}  
\fxsetup{theme=color,mode=multiuser}

\definecolor{color1}{RGB}{46,134,193}
\FXRegisterAuthor{mahdi}{amahdi}{\color{orange}Mahdi}

\definecolor{color7}{RGB}{128,0,128}
\FXRegisterAuthor{yeyuan}{ayeyuan}{\color{red}Yeyuan}

\definecolor{color3}{RGB}{255,128,0}
\FXRegisterAuthor{zihan}{azihan}{\color{blue}Zihan}

\definecolor{color4}{RGB}{150,150,150}
\FXRegisterAuthor{zeyu}{azeyu}{\color{green}Zeyu}

\definecolor{color5}{RGB}{128,128,128}
\FXRegisterAuthor{task}{atask}{\colorbox{color5}{\color{white}Task}}

\newif\ifmynotes
\mynotestrue

\title{Explicit Folded Reed--Solomon and Multiplicity Codes Achieve Relaxed Generalized Singleton Bounds} 
\date{}
\author{
Yeyuan Chen\thanks{Department of EECS, University of Michigan, Ann Arbor. \href{mailto:yeyuanch@umich.edu}{\texttt{yeyuanch@umich.edu}}}  
\and
Zihan Zhang\thanks{Department of Computer Science and Engineering, The Ohio State University. \href{mailto:zhang.13691@buckeyemail.osu.edu}{\texttt{zhang.13691@osu.edu}}}  
}
\author{
\text{Yeyuan Chen}\vspace{6pt}\\{EECS Department}\\University of Michigan, Ann Arbor\\ \href{mailto:yeyuanch@umich.edu}{\texttt{yeyuanch@umich.edu}}  
\and
\text{Zihan Zhang} \vspace{6pt}\\{CSE Department}\\The Ohio State University \\ \href{mailto:zhang.13691@osu.edu}{\texttt{zhang.13691@osu.edu}}
}

\begin{document}

\maketitle
\begin{abstract}
    In this paper, we prove that 
 explicit folded Reed--Solomon (RS) codes and univariate multiplicity codes achieve relaxed generalized
Singleton bounds for list size $L\ge1.$ 
Specifically, we show the following: (1) Any folded RS code of block length $n$ and rate $R$ over the alphabet $\mathbb{F}_q^s$ with distinct evaluation points is $\left(\frac{L}{L+1}\left(1-\frac{sR}{s-L+1}\right),L\right)$ (average-radius) list-decodable for list size $L\in[s]$. 
(2) Any univariate multiplicity code of block length $n$ and rate $R$ over the alphabet $\mathbb{F}_p^s$ (where $p$ is a prime) with distinct evaluation points is $\left(\frac{L}{L+1}\left(1-\frac{sR}{s-L+1}\right),L\right)$ (average-radius) list-decodable for list size $L\in[s]$.

Choosing $s=\Theta(1/\epsilon^2)$ and $L=O(1/\epsilon)$, our results imply that both explicit folded RS codes and explicit univariate multiplicity codes achieve list-decoding capacity $1-R-\epsilon$ with optimal\footnote{Generalized Singleton bound \cite{shangguan2020combinatorial} implies an optimal list size $\lfloor\frac{1-R-\epsilon}{\epsilon}\rfloor$. However, achieving the exact list size needs exponentially large alphabet \cite{BDG24,alrabiah2024ag}, which is impossible in our scenario since we only have polynomial-sized alphabets. Here, the detailed version of our corollary (See \cref{cor:exact_fold}) shows that we can actually achieve list size $L=\lfloor\frac{1-R}{\epsilon}\rfloor$ with $s>\frac{L(L-1)R}{\epsilon L-(1-R-\epsilon)}+L-1$, which makes our list size optimal.} list size $O(1/\epsilon)$. 
This exponentially improves the previous state of the art $(1/\epsilon)^{O(1/\epsilon)}$ established by Kopparty, Ron-Zewi, Saraf, and Wootters (FOCS 2018 or SICOMP, 2023) and Tamo (IEEE TIT, 2024). In particular, our results on folded Reed--Solomon codes fully resolve a long-standing open problem originally proposed by Guruswami and Rudra (STOC 2006 or IEEE TIT, 2008). Furthermore, our results imply the first explicit constructions of $(1-R-\epsilon,O(1/\epsilon))$ (average-radius) list-decodable codes of rate $R$ with polynomial-sized alphabets in the literature.

Our methodology can also be extended to analyze the list-recoverability of folded RS\footnote{Similar list-recoverablility upper bounds are also applicable to univariate multiplicity codes. For brevity, we only focus on folded RS codes in this paper.} codes. We provide a tighter radius upper bound that states folded RS codes cannot be $\left(\frac{L+1-\ell}{L+1}\left(1-\frac{mR}{m-1}\right)+o(1),\ell, L\right)$ list-recoverable where $m=\lceil\log_{\ell}{(L+1)}\rceil>1$. We conjecture this bound is almost tight when $L+1=\ell^a$ for any $a\in\mathbb{N}^{\ge 2}$. To give some evidences, we show folded RS codes over the alphabet $\mathbb{F}_q^s$ are $\left(\frac{1}{2}-\frac{sR}{s-2},2,3\right)$ list-recoverable, which proves the tightness of this bound in the smallest non-trivial special case. As a corollary, our bound states (folded) RS codes cannot be $(1-R-\epsilon, \ell, \ell^{o_R(1/\epsilon)})$ list-recoverable, which refutes the possibility that these codes could achieve list-recovery capacity $\left(1-R-\epsilon, \ell, O\left(\frac{\ell}{\epsilon}\right)\right)$. This implies an intrinsic separation between list-decodability and list-recoverablility of (folded) RS codes.
\end{abstract}
\newpage
\section{Introduction}
An error correcting code $C \subseteq \Sigma^n$ 
is a collection of vectors (codewords) of length $n$ over an alphabet $\Sigma$. In coding theory, a key objective when designing a code $C$ is to ensure that an original codeword $c \in C $ can be recovered from its corrupted version $\tilde{c}\in\Sigma^n$, while also maximizing the size of $C$. One classic case within this framework, known as the unique decoding problem, is to efficiently recover $c\in C$ from any $\tilde{c} \in \Sigma^n$, assuming $c$ and $\tilde{c}$ differ in at most ${\delta}/{2}$ fraction of the positions, where $\delta$ is the relative minimum distance of the code.
\paragraph{List-decoding.} 
As a natural generalization of the unique decoding problem mentioned above, the concept of list-decoding  was introduced independently by Elias \cite{elias1957list} and Wozencraft \cite{wozencraft1958list} in the 1950s. In this setting, the decoder is allowed to output $L \geq 1$ codewords and can potentially correct more than ${\delta}/{2}$ fraction of errors, where $\delta$ represents the relative minimum distance of the code. Since then, list-decoding has found profound applications in theoretical computer science \cite{sudan2000list, vadhan2012pseudorandomness, guruswami2009Unbalanced,ta2012better,goldreich1989hard,  cai1999hardness, goldreich2000chinese} and information theory \cite{elias1991error, rudolf, Blinovski86, Blinovsky1997}.

Formally, a code $C\subseteq\Sigma^n$ over an alphabet $\Sigma$ is defined as  (combinatorially) $(\rho,L)$ list-decodable\footnote{In fact, the results in this paper apply to a stronger notion of list-decodability, called ``average-radius list-decodability.'' Formally, a code $C$ is said to be {$(\rho, L)$ average-radius list-decodable} if there do not exist $y\in\Sigma^n$  and distinct codewords $x_1,\dots,x_{L+1}\in C$ such that $\frac{1}{L+1}\sum_{i=1}^{L+1} \delta(x_i,y)\leq \rho$.} if for every $y\in \Sigma^n$, the Hamming ball centered at $y$ with relative radius $\rho\in [0,1]$ contains at most $L$ codewords from $C$. By the list-decoding capacity theorem \cite[Theorem~7.4.1]{guruswami2019essential}, for $q\geq 2$, $0\leq \rho<1-\frac{1}{q}$, $\eps>0$, and sufficiently large $n$, there exist $(\rho, L)$ list-decodable codes of block length $n$, rate  $R$, alphabet size $q$, and list size $L=O(1/\eps)$ such that
\begin{equation}\label{eq_capacity}
R\geq 1-H_q(\rho)-O(\eps)
\end{equation}
where $H_q(\cdot)$ denotes the $q$-ary entropy function. Codes satisfying \eqref{eq_capacity} are said to achieve the list-decoding capacity. 
When $q\geq 2^{\Omega(1/\eps)}$, condition \eqref{eq_capacity} can be rewritten as $\rho\geq 1-R-O(\eps)$. 

\paragraph{List-decodability of Reed--Solomon codes.} Algebraic codes have played a pivotal role in advancing the study of list-decoding. Among the most significant algebraic codes are Reed–Solomon (RS) codes \cite{reed1960polynomial}, whose codewords are constructed from evaluations of low-degree polynomials. A foundational work of Guruswami and Sudan \cite{sudan1997decoding,guruswami1998improved} provided efficient list-decoding algorithms
for RS codes beyond the unique decoding radius (up to the Johnson bound \cite{johnson1962new}). However, surpassing the Johnson bound presents a substantially greater challenge. 

On the negative side, the work of Ben-Sasson, Kopparty, and Radhakrishnan \cite{ben2009subspace} demonstrated that, over certain finite fields $\mathbb{F}_q$, full-length RS codes are not list-decodable significantly beyond the Johnson bound.  Additionally, other works (e.g. \cite{guruswami2005limits,cheng2007list,gandikota2018np}) have highlighted the challenges of both combinatorial and algorithmic list decoding of RS codes beyond the Johnson radius.

On the positive side, over the past decade,  there has been an exciting line of work \cite{rudra2014every, shangguan2020combinatorial, guo2022improved, ferber2022list, goldberg2022list,brakensiek2023generic,GZ23, alrabiah2023randomly} showing that RS codes with random evaluation points are, with high probability, (combinatorially) list-decodable much beyond the Johnson radius. In particular, Brakensiek, Gopi, and Makam \cite{brakensiek2023generic}
demonstrated that random RS codes are list-decodable (optimally) up
to list-decoding capacity with exponential-sized alphabets. Follow-up works of Guo and Zhang \cite{GZ23}, and Alrabiah,
Guruswami, and Li \cite{alrabiah2023randomly}, improved this result to linear-sized alphabets. 

As a side remark, we note that most of the works in this line of research were inspired by the initial framework of Shangguan and Tamo \cite{shangguan2020combinatorial}, where they introduced the ``generalized Singleton bound'' 
in the context of list-decoding. 
Formally, 
for a linear code of rate $R$ that is $(\rho,L)$ list-decodable, the generalized Singleton bound states that
\begin{equation}\label{for:GSB}
    \rho\leq\frac{L}{L+1}\left(1-R\right).
\end{equation}
We also have a $\epsilon$-slack version of the generalized Singleton bound, $\frac{L}{L+1}(1-R-\epsilon)$, introduced in \cite{GZ23}.

Nevertheless, explicit constructions of Reed--Solomon codes (with list-decoding algorithms) that surpass the Johnson radius remain a widely open problem, which is beyond the scope of this paper. 

\paragraph{List-decodability of folded Reed--Solomon codes.} Folded RS codes, initially introduced by Krachkovsky in \cite{krachkovsky2003reed}, are a simple variant of RS codes. Formally, given any parameters $s,n,k>0$, a finite field $\mathbb{F}_q$, where $|\mathbb{F}_q|>sn>k$, a generator $\gamma$ of $\mathbb{F}^\times_q$, and any sequence of $n$ elements $\alpha_1,\alpha_2,\dots,\alpha_n\in\F_q$, the corresponding $(s,\gamma)$-folded RS code over the alphabet $\mathbb{F}_q^s$ with block length $n$ and rate $R=\frac{k}{sn}$ is defined as 
\[\mathsf{FRS}^{(s,\gamma)}_{n,k}(\alpha_1,\alpha_2,\dots,\alpha_n):=\bigg\{\mathcal{C}(f):f\in\mathbb{F}_q[x],\text{ }\deg f<k\bigg\}\subseteq\left(\mathbb{F}_q^s\right)^n,
\]
where $\mathcal{C}(f):=(F_1,F_2,\dots,F_n)\in\left(\mathbb{F}_q^s\right)^n$, with $F_i:=\Bigl(f(\alpha_i),f(\gamma\alpha_i),\dots,f(\gamma^{s-1}\alpha_i)\Bigl)$, denotes the encoder of this code. Note that RS codes correspond to the special case of $s=1$.

Building on the breakthrough of Parvaresh and Vardy \cite{parvaresh2005correcting}, the seminal work by Guruswami and Rudra \cite{Guru06,guruswami2008explicit}
demonstrated that folded RS codes can be efficiently list-decoded up to list-decoding capacity with a polynomial-sized list. More precisely, they showed that $(s,\gamma)$-folded RS codes are $(1-R-\epsilon,n^{O(1/\epsilon)})$ list-decodable with $s\approx 1/\epsilon^2$. A natural question was then posed by Guruswami and Rudra in \cite{Guru06, guruswami2008explicit} regarding the potential reduction of the list size.
\begin{open}[Guruswami--Rudra \cite{Guru06,guruswami2008explicit}]\label{open}
It remains an open question to reduce this list size $n^{O(1/\epsilon)}$, given that existential random coding arguments work with a list size of $O(1/\epsilon)$. 
\end{open}
Over the past (almost) two decades, there have been numerous successful attempts \cite{gur11,guruswami2013linear,kopparty2018improved,kopparty2023improved,tamo2024tighter}
\text{ }to the above open problem. However, it still {\bf remains open} whether the list size of explicit folded RS codes can be
brought down all the way to $O(1/\epsilon)$, which exactly matches the list size achieved by the probabilistic
method. 

Now we provide a brief summary of the aforementioned attempts. The first significant step after the work of Guruswami and Rudra was made by Guruswami and Wang \cite{gur11,guruswami2013linear}. They introduced a linear-algebraic list-decoding algorithm and made the surprising discovery that all of the messages in the list are contained within a $\mathbb{F}_q$-linear subspace of dimension $O(1/\epsilon)$. 

Although Guruswami and Wang did not directly improve the list size bound of folded RS codes, they constructed folded RS subcodes based on a pseudorandom
object called ``subspace evasive sets'' and showed that the list size can be reduced to $O(1/\epsilon)$ for randomized folded RS subcodes. In a follow-up work, Dvir
and Lovett \cite{dvir2012subspace} gave an explicit construction of ``subspace evasive sets,'' which led to a list size of $\left(1/\epsilon\right)^{O(1/\epsilon)}$ for explicit folded RS subcodes. Notably, these ``subspace evasive sets'' are inherently non-linear, making the resulting codes (folded RS subcodes) non-linear as well.

In 2018, an important work by Kopparty, Ron-Zewi, Saraf, and Wootters \cite{kopparty2018improved,kopparty2023improved} demonstrated that the list size for explicit folded RS codes can be reduced to a constant $\left(1/\epsilon\right)^{O(1/\epsilon)}$, independent of the code length $n$. A follow-up work by Tamo \cite{tamo2024tighter} provided a more careful analysis, further improving the list-size to $\left(1/\epsilon\right)^{4/\epsilon}$. Although these two works have partially resolved the open problem posed by Guruswami and Rudra, the current list size remains exponential in $1/\epsilon$. 

Regarding Tamo's recent work \cite{tamo2024tighter}, it is worth noting that for the special case of list size $L=2$, he showed that explicit folded RS codes approximately achieve the generalized Singleton bound for list decoding. However, based on our understanding, Tamo's approach, which builds on the techniques developed in \cite{kopparty2018improved, kopparty2023improved}, does not appear to generalize readily to any $L\geq 3$.

Surprisingly, despite numerous other constructions of capacity-achieving codes over the past two decades \cite{guruswami2013list, kopparty15, guo2021efficient, guruswami2022optimal, bhandari2023ideal, berman2024explicit, brakensiek2024ag, guo2024random}, we are not aware of any \emph{explicit} constructions of $(1-R-\epsilon,O(1/\epsilon))$ list-decodable codes of rate $R$ over either polynomial-sized or constant-sized alphabets, even for non-linear codes. In this paper, we show that explicit folded RS codes achieve the above list-decodability, providing an affirmative answer to \Cref{open}.

\paragraph{List-decodability of univariate multiplicity codes.} Univariate multiplicity codes were first introduced in \cite{rt97_mult}, with a multivariate variant later developed in \cite{KSY14}. These codes are defined in terms of Hasse derivatives. Given a finite field $\mathbb{F}_q, i\in\mathbb{N}$, and a polynomial $f(X)$, the $i$-th Hasse derivative $f^{(i)}(X)$ is defined as the coefficient of $Z^i$ in the expansion
\[f(X+Z)=\sum_{i\in\mathbb{N}}f^{(i)}(X)Z^i.\]
Using Hasse derivatives, we can formally define the univariate multiplicity code. Given parameters $s,n,k>0$, a finite prime field\footnote{For simplicity, we only consider univariate multiplicity codes defined over a finite prime 
field.} $\mathbb{F}_p$, where $|\mathbb{F}_p|\ge n$, and any sequence of $n$ elements  $\alpha_1,\dots,\alpha_n\in\mathbb{F}_p$, the corresponding order-$s$ univariate multiplicity code over the alphabet $\mathbb{F}^s_p$ with block length $n$ and rate $R=\frac{k}{sn}$ is defined as 
\[\mathsf{MULT}_{n,k}^{(s)}(\alpha_1,\alpha_2,\dots,\alpha_n):=\bigg\{\mathcal{M}(f):f\in\mathbb{F}_p[x],\text{ }\deg f<k\bigg\}\subseteq\left(\mathbb{F}_p^s\right)^n,\]
where $\mathcal{M}(f):=(F_1,F_2,\dots,F_n)\in\left(\mathbb{F}_p^s\right)^n$, with $F_i:=\Bigl(f(\alpha_i),f^{(1)}(\alpha_i),\dots,f^{(s-1)}(\alpha_i)\Bigl)$, denotes the encoder of this code. Similarly, RS codes correspond to the special case of $s=1$.

Parallel to the results on list-decoding folded RS codes, it was shown in \cite{kopparty15,guruswami2013linear} that explicit univariate multiplicity codes over $\mathbb{F}_p^s$ achieves list-decoding capacity with a polynomial list size $L=n^{O(1/\epsilon)}$. Similar reductions in list size have been made for explicit univariate multiplicity codes, with a list size of $(1/\epsilon)^{\frac{4}{\epsilon}\left(1+\frac{Rn}{p}\right)}$ by Kopparty, Ron-Zewi, Saraf, Wootters \cite{kopparty2018improved, kopparty2023improved} and Tamo \cite{tamo2024tighter}.
\text{ }A similar question can be posed for explicit univariate multiplicity codes: whether it is possible to reduce the list size from $(1/\epsilon)^{\frac{4}{\epsilon}\left(1+\frac{Rn}{p}\right)}$ to $O(1/\epsilon)$. This question will also be addressed in this paper. 

\paragraph{List-recovery.} 
List-recovery is a natural generalization of list-decoding, which is more challenging and has broader applications in areas such as pseudorandomness \cite{guruswami2009Unbalanced} and algorithms \cite{dw22}. Formally, for any code $\mathcal{C}\subseteq \Sigma^n$ over an alphabet $\Sigma$, we say $\mathcal{C}$ is $(\rho, \ell, L)$ list-recoverable if, for any product set $S = S_1 \times \cdots \times S_n \in \binom{\Sigma}{\leq \ell}^n$ with $|S_1|,\dots,|S_n|\leq \ell$, there are at most $L$ codewords $c \in \mathcal{C}$ such that
\[\mathrm{dist}(c,S_1\times\cdots\times S_n)\leq\rho n\]
where $\mathrm{dist}(c,S)$ denotes the number of indices $i\in[n]$ such that $c_i\notin S_i$. 

A natural extension of the Johnson bound for list-recovery (see \cite{gs01}) shows that any MDS code is $(1 - \sqrt{\ell R}, \ell, \poly(n))$ list-recoverable. Furthermore, similar to list-decoding, there is a list-recovery capacity (see \cite{gi01, guthesis, resthesis}). 
Let $q \geq \ell \geq 1$, $0 < \rho < 1 - \frac{\ell}{q}$, and $\epsilon > 0$. 
As stated in \cite[Theorem 2.4.12]{resthesis}, for sufficiently large $n$, random codes of block length $n$, rate $R$, and alphabet size $q$ are $(\rho, \ell, O\left(\frac{\ell}{\epsilon}\right))$ list-recoverable with high probability, where \begin{equation}\label{recovercapacity} R \geq 1 - H_{q, \ell}(\rho) - \epsilon, \end{equation} and $H_{q, \ell}(\rho) := \rho \log_q{\left(\frac{q - \ell}{\rho}\right)} + (1 - \rho) \log_q{\left(\frac{\ell}{1 - \rho}\right)}$ is the $(q, \ell)$-ary entropy function defined in \cite[Definition 2.4.9]{resthesis}. Codes satisfying \eqref{recovercapacity} are said to achieve list-recovery capacity\footnote{There is another slightly different formulation of list-recovery capacity, $R \geq 1 - H_{q/\ell}(\rho) - \epsilon$, in the literature (see \cite{aw18}). In this paper, we use the threshold \eqref{recovercapacity} from \cite{resthesis}, as it behaves similarly to list-decoding capacity in the list-recovery regime in most general cases.}. When $q \geq 2^{\Omega\left(\log{\ell}/\epsilon\right)}$, condition \eqref{recovercapacity} simplifies to $\rho \geq 1 - R - \epsilon$. We refer to this variant as \emph{large-alphabet list-recovery capacity} in this paper.

On the negative side, the best previous result from \cite{gst22b} shows that for any $\ell \geq 2$ and $\epsilon > 0$, any code with sufficiently large block length $n$ and constant rate $R$ cannot be $\left(1 - R - \epsilon, \ell, \frac{\ell(1 - R)}{\epsilon} - 2\right)$ list-recoverable. The list size in this result matches the large-alphabet list-recovery capacity up to a constant factor. 

\paragraph{List-recoverability of folded RS and univariate multiplicity codes.} There is a long line of exciting results \cite{Guru06,guruswami2013linear,kopparty2018improved,kopparty2023improved,lp20,guo2022improved,goldberg2022list,tamo2024tighter} on the list-recoverability of randomly punctured and explicit RS codes, as well as folded RS codes, beyond the Johnson radius.  In particular, for explicit folded RS codes, \cite{Guru06} first established their $(1-R-\epsilon, \ell, \poly(n))$ list-recoverability. Later, \cite{guruswami2013linear} refined this result by showing that explicit $(s,\gamma)$-folded RS codes are $\left(\frac{1}{L+1}\left(L+1-\ell-\frac{sLR}{s-L+1}\right), \ell, \poly(n)\right)$ list-recoverable, where the list of candidate messages is contained in an affine subspace of dimension at most $L-1$. Building on this and introducing new techniques, \cite{kopparty2018improved,kopparty2023improved,tamo2024tighter} reduced the list size to a constant. Specifically, they demonstrated that explicit folded RS codes are $\left(1-R-\epsilon, \ell, \left(\frac{\ell}{\epsilon}\right)^{O(\frac{1+\log{\ell}}{\epsilon})}\right)$ list-recoverable. A similar bound applies to univariate multiplicity codes using a similar analysis.

In this work, given parameters $\ell$ and $L$, we present a new upper bound on the radius $\rho$, beyond which any folded RS codes cannot be $(\rho, \ell, L)$ list-recoverable. We also provide evidence that this upper bound is close to optimal. Our new bound is exponentially tighter than the previous result from \cite{gst22b}, and the proof strategy can be adapted to demonstrate similar results for univariate multiplicity codes.

\subsection{Main results}
In this work, we fully resolve the open question posed by Guruswami and Rudra in \cite{Guru06,guruswami2008explicit} (see \cref{open}) by demonstrating that both explicit folded RS codes and explicit univariate multiplicity codes are (average-radius) list-decodable up to capacity with an optimal list size of $O(1/\epsilon)$. More broadly, we show that these two code families achieve a relaxed version of the generalized Singleton bound in the context of list-decoding.

As a side remark, our results imply that folded RS codes and univariate multiplicity codes are explicit constructions of $(1-R-\epsilon, O(1/\epsilon))$ list-decodable codes with rate $R$ over polynomial-sized alphabets, making them the first explicit constructions to achieve such parameters. 

In terms of list-recoverability, given $\ell$ and $L$, we establish an improved upper bound on the radius $\rho$ beyond which folded RS codes cannot be $(\rho, \ell, L)$ list-recoverable. As a corollary, our bound implies that folded RS codes cannot be $(1-R-\epsilon, \ell, \ell^{o_R(1/\epsilon)})$ list-recoverable. This result suggests that the list-recoverability bounds presented in \cite{kopparty2018improved,kopparty2023improved,tamo2024tighter} are close to optimal and that folded RS codes exhibit significantly worse list sizes than random codes in the context of list-recovery. Our new bound provides an exponentially tighter lower bound on the list size compared to the previous bound of $\left(1-R-\epsilon, \ell, \frac{\ell(1-R)}{\epsilon}-2\right)$ from \cite{gst22b}. Using similar arguments, these results can also be applied to univariate multiplicity codes.

\paragraph{List-decodability of folded Reed--Solomon codes.}  Our first main theorem states that any folded Reed--Solomon code associated with an ``appropriate'' tuple of field elements satisfies a relaxed version of the generalized Singleton bound. 
\begin{definition}[Appropriate evaluation points]
Given a finite field $\mathbb{F}_q$, $s \geq 1$, and a generator $\gamma$ of the cyclic group $\mathbb{F}_q^\times$, a tuple $(\alpha_1, \dots, \alpha_n) \in \mathbb{F}_q^n$ is called \emph{appropriate} if the $sn$ elements $\gamma^i\alpha_j$ are distinct, where $i\in\{0,\dots,s-1\}$ and $j\in[n]$. Throughout this paper, the choice of $\gamma$ will be clear from the context. 
\end{definition}
\begin{theorem}\label{thm:main}
For any integers $s,n,L\geq 1$, $k\in [n]$, generator $\gamma$ of $\mathbb{F}^{\times}_q$ and appropriate $(\alpha_1,\alpha_2,\dots,\alpha_n)\in\F_q^n$, the code $\mathsf{FRS}_{n,k}^{(s,\gamma)}(\alpha_1,\alpha_2,\dots,\alpha_n)$ over the alphabet $\mathbb{F}_q^s$ is $\left(\frac{L}{L+1}\left(1-\frac{sR}{s-L+1}\right),L\right)$ list-decodable\footnote{\cref{thm:main} holds even when list decodability is replaced by average-radius list-decodability. The same is true for \cref{cor:main}, \cref{thm:main_mul}, and \cref{cor:main_mul}.}.
\end{theorem}
An interesting observation is that the bound $\frac{L}{L+1}\left(1-\frac{sR}{s-L+1}\right)$ matches the decoding radius given in \cite[Theorem 7]{gur11} given the parameter $L$.

As a corollary, we prove that any folded RS code associated with an appropriate tuple $(\alpha_1,\alpha_2,\dots,\alpha_n)$ achieves list-decoding capacity with an optimal list size.
\begin{corollary}[Informal, see \cref{cor:exact_fold}]\label{cor:main}
Let $\epsilon>0$, $n \geq 1$, $k\in [n]$, generator 
$\gamma$ of $\mathbb{F}^{\times}_q$,  $L=O(1/\epsilon)$, and $s=\Theta(1/\epsilon^2)$. Let $(\alpha_1,\alpha_2,\dots,\alpha_n)\in\F_q^n$ be an appropriate tuple.  Then the folded RS code $\mathsf{FRS}_{n,k}^{(s,\gamma)}(\alpha_1,\alpha_2,\dots,\alpha_n)$ over the alphabet $\mathbb{F}_q^s$ is $\big(1-R-\epsilon,O(1/\epsilon)\big)$ list-decodable.
\end{corollary}
\paragraph{List-decodability of univariate multiplicity codes.} Using a similar proof strategy, we establish the same list-decodability for univariate multiplicity codes over $\mathbb{F}_p^s$, where $p$ is a prime number.
\begin{theorem}\label{thm:main_mul}
Let $p$ be a prime number. For any integers $s,n,L\geq 1$, $k\in [n]$, and distinct $\alpha_1,\alpha_2,\dots,\alpha_n\in\F_p$, the code $\mathsf{MULT}_{n,k}^{(s)}(\alpha_1,\alpha_2,\dots,\alpha_n)$ over the alphabet $\mathbb{F}_p^s$ is $\left(\frac{L}{L+1}\left(1-\frac{sR}{s-L+1}\right),L\right)$ list-decodable.
\end{theorem}
\begin{corollary}[Informal, see \cref{cor:exact_mult}]\label{cor:main_mul}
Let $p$ be a prime number, $\epsilon>0$, $n \geq 1$, $k\in [n]$, $L=O(1/\epsilon)$, and $s=\Theta(1/\epsilon^2)$. For any distinct $\alpha_1,\alpha_2,\dots,\alpha_n\in\F_q$, the order-$s$ univariate multiplicity code $\mathsf{MULT}_{n,k}^{(s)}(\alpha_1,\alpha_2,\dots,\alpha_n)$ over the alphabet $\mathbb{F}_p^s$ is $\big(1-R-\epsilon,O(1/\epsilon)\big)$ list-decodable.
\end{corollary}
\paragraph{More general results for list-decoding.} Based on the connection with the notion of ``subspace designs'' introduced in \cite{guruswami2013list}, our techniques extend to a more general result on codes achieving list-decoding capacity. Specifically, we show that any ``strong subspace designable codes'' (see \cref{def:designcodes})  exhibit near-optimal (average-radius) list-decodability. Our proofs on folded RS and univariate multiplicity codes are two special cases of this generalization, and we can actually wrap up \cref{sec:frs} and \cref{sec:mul} to make them more modular using the ``subspace design'' language. However, To keep our proofs (\cref{sec:frs}) self-contained, we defer this generalization (see \cref{thm:generalize}) and modular formulation to \cref{append:generalize}.

\paragraph{List-recoverability of folded Reed–Solomon codes.}
We further extend our new framework for list-decoding to investigate the list-recoverability of folded RS codes. For any $\ell$ and $L$, we provide a bound on the radius $\rho$ indicating that any folded RS code associated with an appropriate tuple of field elements cannot be $(\rho, \ell, L)$ list-recoverable if $\rho$ is too large.
\begin{theorem}\label{thm:bound_clean}
Let $s,n,L\ge 1$, $k\in [n]$, $2\leq \ell\leq L,q\ge\ell$, generator $\gamma$ of $\mathbb{F}^{\times}_q$ and $m=\lceil\log_{\ell}(L+1)\rceil>1$. Suppose $R\leq \frac{m-1}{m}$ and $\frac{k-1}{s}\ge m$. If a folded Reed--Solomon code $\mathsf{FRS}^{(s,\gamma)}_{n,k}(\alpha_1,\alpha_2,\dots,\alpha_n)$ of rate $R=\frac{k}{sn}$ with appropriate evaluation points in $\mathbb{F}_q$ is $(\rho,\ell,L)$ list-recoverable, then 
    \[\rho\leq\frac{L+1-\ell}{L+1}\left(1-\frac{mR}{m-1}\right)+\frac{5}{n},\]  
\end{theorem}
When $L+1=\ell^m$ for any $m, \ell\ge 2$, we have $\frac{m(L+1-\ell)}{m-1}>\frac{sL}{s-L+1}$ for sufficiently large $s$. Therefore, \cref{thm:bound_clean} provides a tighter upper bound, which disproves a previously held conjecture that explicit folded RS codes are $(\rho, \ell, L)$ list-recoverable for the radius $\rho=\frac{1}{L+1}\left(L+1-\ell-\frac{sLR}{s-L+1}\right)$, which was originally proposed in \cite{gur11,guruswami2013linear}. Additionally, we conjecture that our bound is nearly tight when $L+1=\ell^a$ for any integer $a\geq 2$. Concretely, we believe any  folded RS codes with appropriate evaluation points achieves the $\epsilon$-relaxed radius bound when the folding parameter $s$ is large enough (see \cref{conj:tight}). As an evidence, we prove the tightness of our bound when $(\ell,L)=(2,3)$, as stated by \cref{thm:special} below. It is worth noting that $(\ell,L)=(2,3)$ is the smallest ``non-trivial'' list-recoverability case that cannot be ``trivially'' derived from known list-decodability results. 
\begin{theorem}\label{thm:special}
Let $n,L\ge 1$, $k\in [n]$, $q>n$, $s\geq 3$,
and generator $\gamma$ of  $\mathbb{F}_q^\times$.
The folded RS code $\mathsf{FRS}_{n,k}^{(s,\gamma)}(\alpha_1,\alpha_2,\dots,\alpha_n)$ with appropriate evaluation points in $\mathbb{F}_q$ is $\left(\frac{1}{2}-\frac{sR}{s-2}, 2, 3\right)$ list-recoverable.
\end{theorem}
In order to compare our bound  \cref{thm:bound_clean} with list-recovery capacity, we may rewrite \cref{thm:bound_clean} in the following form to get more intuitions. 
\begin{corollary}\label{cor:rewrite_bound}
For any constants $0<R<1, \ell\ge 2$, $0<\epsilon<\frac{R(1-R)}{4}, s\ge 1$ and generator $\gamma$ of $\mathbb{F}^{\times}_q$ if $k,n$ are sufficiently large, then any rate $R=\frac{k}{sn}$ folded Reed--Solomon code $\mathsf{FRS}^{(s,\gamma)}_{n,k}(\alpha_1,\dots,\alpha_n)$ with appropriate evaluation points in $\mathbb{F}_q$ cannot be $(1-R-\epsilon, \ell, \ell^{\frac{R}{2\epsilon}-1}-1)$ list-recoverable.
\end{corollary}
On the negative side, when $\ell$ is a constant, our bound \cref{cor:rewrite_bound} refutes the hope that folded RS codes could achieve the large-alphabet list-recovery capacity $\left(1-R-\epsilon, \ell, O(\frac{\ell}{\epsilon})\right)$ that is achieved by random codes. \cref{cor:rewrite_bound} exponentially improves the previous best known list size lower bound $\left(1-R-\epsilon, \ell, \frac{\ell (1-R)}{\epsilon}-2\right)$ from \cite{gst22b}.  It is also worth noting that \cite{kopparty2018improved,kopparty2023improved,tamo2024tighter} has proved the $\left(1-R-\epsilon, \ell, \left(\frac{\ell}{\epsilon}\right)^{O(\frac{1+\log{\ell}}{\epsilon})}\right)$ list-recoverability of folded RS codes, 
so \cref{cor:rewrite_bound} actually implies that the exponential dependency on ${1}/{\epsilon}$ in the list size cannot be removed. Besides, since we have shown these codes achieve list-decoding capacity, \cref{cor:rewrite_bound} implies an intrinsic separation between list-decoding and list-recovery. More specifically, folded RS 
 codes are as good as random codes in terms of list-decodability, but they are far worse than random codes when considering list-recoverability.

As a side remark, RS codes are just a special case of folded RS codes when $s=1$, so \cref{thm:bound_clean,cor:rewrite_bound} also apply to them. We illustrate this set of results about list-recoverability in \cref{sec:list_recover}. 
\begin{remark*}
   While these list-recoverable upper bounds are also applicable to univariate multiplicity codes, we will focus solely on folded RS codes for the sake of brevity.
\end{remark*}

\paragraph{Concurrent work.} 
After our work was completed, we learned about a concurrent and independent work by Shashank Srivastava, which was first announced in his thesis \cite{srivastava2024continuous} \footnote{See also their updated version \cite{srivastava24},  scheduled for publication at SODA 2025.} 
around the same time. 
Using different techniques, Srivastava obtained a similar but weaker result on the list-decodability of folded RS codes.
Specifically, we prove the $\left(\frac{L}{L+1}\left(1-\frac{sR}{s-L+1}\right),L\right)$ (average-radius) list-decodablility for folded RS codes, while Srivastava \cite{srivastava2024continuous} shows the $\left(\frac{L}{L+1}\left(1-\frac{sR}{s-L+1}\right),L^2\right)$ list-decodablility. Consequently, our result fully resolves \cref{open} of Guruswami--Rudra with the optimal list size $O(1/\epsilon)$, while Srivastava's result implies only a quadratic list size $O(1/\epsilon^2)$. \cref{sec:loss} is the major technical difference that allows our paper to outperform Srivastava's result and reduce the list size to be optimal.

\vspace{14pt}
Prior to presenting formal proofs for the aforementioned theorems, we introduce essential notation and outline the organization of the paper below for clarity and convenience.
\paragraph{Notation.} Throughout this paper, unless stated otherwise, we use $f:=f(X)=\sum_{i=0}^{k-1}a_iX^i$ to denote a polynomial over $\mathbb{F}_q$ with degree at most $k-1$, and $\vec{f}\in\mathbb{F}^k_q$ to denote the corresponding coefficient vector $\left(a_0,\dots,a_{k-1}\right)^T$.  The parameters $k$ and $q$ will be clear from the context. For any vector $y\in\left(\mathbb{F}^s_q\right)^n$ and $i\in[n]$, we use $y[i]\in\mathbb{F}^s_q$ to denote its $i$-th entry. The symbol $\gamma$ always denotes a generator of the cyclic group $\mathbb{F}_q^{\times}$. Unless otherwise stated, the list size $L$ and the folding parameter $s$ are positive integers. For any finite set $A$, we use $2^A$ to denote the power set of $A$, $\binom{A}{k}$ to denote the set of all subsets of $A$ with size $k$, and $\binom{A}{\leq k}:=\bigcup_{t\leq k}\binom{A}{t}$. For any undirected graph $G$, we use $E(G)$ to denote its edge-set.

\paragraph{Paper organization.} In \cref{sec:frs}, we prove our results on the list-decodability of folded RS codes (\cref{thm:main} and \cref{cor:main}). By extending our method from \cref{sec:frs}, we derive a tighter upper bound on the achievable radius for the list-recoverability of folded RS codes in \cref{sec:list_recover}. Based on the framework established in \cref{sec:frs}, we complete the proofs of our list-decoding results for univariate multiplicity codes (\cref{thm:main_mul} and \cref{cor:main_mul}) in \cref{sec:mul}. Additionally,  using the ``subspace design'' language, we provide a more general theorem that all ``strong subspace designable codes'' achieve list-decoding capacity in \cref{append:generalize}. It also gives a more modular but black-box presentation of our proofs in \cref{sec:frs,sec:mul}.

\section{Folded Reed--Solomon Codes Achieve Relaxed Generalized Singleton Bounds 
}\label{sec:frs}
In this section, we prove our main results on the list-decodability of folded RS codes (\cref{thm:main} and \cref{cor:main}). 
\subsection{Geometric agreement hypergraphs and geometric polynomials} 
The notion of ``agreement hypergraphs'' was introduced by Guo, Li, Shangguan, Tamo, and Wootters \cite{guo2022improved} to formalize a ``hypergraph Nash-Williams--Tutte conjecture'' and to study the list-decodability of random RS codes. In this section,
we introduce an expanded notion, called ``geometric agreement hypergraphs,'' which extends agreement hypergraphs by incorporating additional geometric information: each vertex is represented by a vector in the linear space $\mathbb{F}_q^k$. 

We also introduce another notion called ``geometric polynomials'' in this section. A geometric polynomial is invariant under different choices of basis and, together with the geometric agreement hypergraph, plays a crucial role in our proof.

\subsubsection{Geometric agreement hypergraphs}
We will now formally define the notion of ``geometric agreement hypergraph'' below. Each vertex of this hypergraph will represent a vector, and therefore each of its hyperedges is related to the geometric positions of vectors in it. It enables us to inspect list-decodability from a geometric perspective.
\begin{definition}[Geometric agreement hypergraph based on FRS codes]\label{def:agr_graph}
Let $\gamma$ be a generator of $\mathbb{F}^{\times}_q$ Given a $(s,\gamma)$-folded Reed--Solomon code $\mathsf{FRS}^{(s,\gamma)}_{n,k}(\alpha_1,\dots,\alpha_n)\subseteq \left(\mathbb{F}^s_q\right)^n$ where $(\alpha_1,\dots,\alpha_n)$ is an appropriate sequence, a received word $\vec{y}\in(\mathbb{F}^s_q)^n$, and $\ell$ vectors $\vec{f_1},\vec{f_2},\dots,\vec{f_\ell}\in \mathbb{F}_q^k$, we define the geometric agreement hypergraph $(\mathcal{V},\mathcal{E})$ with vertex set $\mathcal{V}:=\left\{\vec{f_1},\vec{f_2},\dots,\vec{f_\ell}\right\}$ and a tuple of $n$ hyperedges $\mathcal{E}:=\{e_1,e_2,\dots,e_n\}$, where
$e_i:=\left\{\vec{f_j}\in\mathcal{V} : \vec{y}[i]=\mathcal{C}(f_j)[i]\right\}.$ 

Additionally, given any subset $\mathcal{H}\subseteq \mathcal{V}$, we define $(\mathcal{H},\mathcal{E}|_{\mathcal{H}})$ as the geometric agreement sub-hypergraph of $\left(\mathcal{V},\mathcal{E}\right)$ restricted on $\mathcal{H}$, where $\mathcal{E}|_{\mathcal{H}}:=\{e_1|_{\mathcal{H}},\dots,e_n|_{\mathcal{H}}\}$ and $e_i|_{\mathcal{H}}:=\{\vec{f}_t:\vec{f}_t\in e_i\cap \mathcal{H}\}$.
\end{definition}
We will then define the notion of ``affine dimension'' for a set of vectors, which characterizes how much information we can get from a hyperedge consisting of vectors in a geometric agreement hypergraph. The affine dimension of a set of vectors depends solely on the relative geometric positions of these vectors.  

\begin{definition}[Affine dimension] 
    For any 
$\vec{f}_1,\vec{f}_2,\dots,\vec{f}_m\in\mathbb{F}_q^k$, we define
\[\widetilde{\dim}_{\mathbb{F}_q}\left(\vec{f}_1,\vec{f}_2,\dots,\vec{f}_m\right):=\min\bigg\{r-1:\left\{\vec{f}_1,\vec{f}_2,\dots,\vec{f}_m\right\}\subseteq\mathrm{SP}\left(\vec{f}_{i_1},\vec{f}_{i_2},\dots,\vec{f}_{i_r}\right)\text{ for some }r\in[m]\bigg\},\]
where $\mathrm{SP}\left(\vec{f}_{1},\vec{f}_{2},\dots,\vec{f}_{s}\right):=\big\{a_1\vec{f}_{1}+a_2\vec{f}_{2}+\dots+a_s\vec{f}_{s}:a_1,a_2,\dots,a_s\in\F_q \text{ and }\sum_{i\in[s]}a_i=1\big\}\subseteq\mathbb{F}_q^k.$
\end{definition}
\begin{remark}
    Intuitively, $\widetilde{\dim}_{\mathbb{F}_q}\left(\vec{f}_1,\dots,\vec{f}_m\right)$ means the dimension of the smallest affine subspace that contains all these vectors.
\end{remark}
Since a hyperedge is essentially a set of vectors, we can also define the affine dimension of it.
\begin{definition}[Affine dimension of hyperedges]
Using notations in \Cref{def:agr_graph}, for any geometric agreement hypergraph $(\mathcal{V},\mathcal{E})$ and its hyperedge $e\in\mathcal{E}$, let $e=\left(\vec{f}_{i_1},\dots,\vec{f}_{i_t}\right)$ where $t=|e|$, we can define the affine dimension of the hyperedge $e$ by
$\widetilde{\dim}_{\mathbb{F}_q}\left(e\right):=\widetilde{\dim}_{\mathbb{F}_q}\left(\vec{f}_{i_1},\dots,\vec{f}_{i_t}\right)$.
\end{definition}
The following fact about affine dimension will be frequently used.
\begin{fact}\label{fact:transform_tildedim}
For any $\vec{f}_1,\dots,\vec{f}_m\in\mathbb{F}^k_q$ and $i\in[m]$, it follows that 
\[\widetilde{\dim}_{\mathbb{F}_q}\left(\vec{f}_1,\dots,\vec{f}_m\right)=\dim_{\mathbb{F}_q}\operatorname{Span}_{\mathbb{F}_q}\left\{\vec{f}_1-\vec{f}_i,\dots,\vec{f}_m-\vec{f}_i\right\}.\]
\end{fact}
\begin{proof}
Given any $i\in[m]$, let $\vec{f}^\prime_j=\vec{f}_j-\vec{f}_i,j\in[m]$ and $V=\Span_{\mathbb{F}_q}\left\{\vec{f}^\prime_1,\dots,\vec{f}^\prime_m\right\}$, we prove the following two directions:
\paragraph{$\widetilde{\dim}_{\mathbb{F}_q}\left(\vec{f}_1,\dots,\vec{f}_m\right)\leq\dim_{\mathbb{F}_q}V$.} Let $\ell=\dim_{\mathbb{F}_q}V$ and $\left\{\vec{f}^\prime_{u_j}\right\}_{j\in[\ell]}$ a basis of $V$, then for any $t\in[m]$, it follows that $\vec{f}_t=\vec{f}^\prime_t+\vec{f}_i=\left(\sum_{j\in[\ell]}a_j\vec{f}^\prime_{u_j}\right)+\vec{f}_i$ where $a_j\in\mathbb{F}_q,j\in[\ell]$. We can rewrite it as $\vec{f}_t=\left(\sum_{j\in[\ell]}a_j\vec{f}_{u_j}\right)+\left(1-\sum_{j\in[\ell]}a_j\right)\vec{f}_i$. Note that $u_j\neq i$ for any $j\in[\ell]$, we know that $\vec{f}_t\in\mathrm{SP}\left(\vec{f}_{u_1},\dots,\vec{f}_{u_\ell},\vec{f}_i\right)$. Since it holds for any $t\in[m]$, there must be $\widetilde{\dim}_{\mathbb{F}_q}\left(\vec{f}_1,\dots,\vec{f}_m\right)\leq\ell$.
\paragraph{$\widetilde{\dim}_{\mathbb{F}_q}\left(\vec{f}_1,\dots,\vec{f}_m\right)\ge\dim_{\mathbb{F}_q}V$.} Let $\ell=\widetilde{\dim}_{\mathbb{F}_q}\left(\vec{f}_1,\dots,\vec{f}_m\right)$, then there exists $u_1,\dots,u_{\ell+1}\in[m]$ such that $\left\{\vec{f}_1,\vec{f}_2,\dots,\vec{f}_m\right\}\subseteq\mathrm{SP}\left(\vec{f}_{u_1},\dots,\vec{f}_{u_{\ell+1}}\right)$. We also define $\vec{g}_j=\vec{f}_j-\vec{f}_{u_{\ell+1}},j\in[m]$ and $W=\Span_{\mathbb{F}_q}\left\{\vec{g}_{u_1},\dots,\vec{g}_{u_{\ell}}\right\}$. For any $t\in[m]$, since we can write $\vec{f}_t=\left(\sum_{j\in[\ell]}a_j\vec{f}_{u_j}\right)+\left(1-\sum_{j\in[\ell]}a_j\right)\vec{f}_{u_{\ell+1}}$ where $a_j\in\mathbb{F}_q$ and $j\in[\ell]$, it follows that $\vec{g}_t=\sum_{j\in[\ell]}a_j\vec{g}_{u_j}$. This implies $\left\{\vec{g}_1,\dots,\vec{g}_m\right\}\subseteq W$. Therefore, for any $t\in[m]$, we have $\vec{f}^\prime_t=\vec{g}_t-\vec{g}_{i}\in W-W=W$, and we conclude $\dim_{\mathbb{F}_q}V\leq \dim_{\mathbb{F}_q} W\leq\ell$.
\end{proof}
Following the approach of Shangguan and Tamo \cite{shangguan2020combinatorial}, we define the ``weight'' of geometric agreement hypergraphs, sub-hypergraphs, and hyperedges.
\begin{definition}[Weight]
Given a geometric agreement hypergraph $(\mathcal{V},\mathcal{E})$ where $\mathcal{E}=\{e_1,\dots,e_n\subseteq \mathcal{V}\}$, we define the weight $\mathrm{wt}(\mathcal{V},\mathcal{E}):=\sum^n_{i=1}\mathrm{wt}(e_i)$, where $\mathrm{wt}(e_i):=\max\bigg(|e_i|-1,0\bigg)$.
\end{definition}
\subsubsection{Geometric polynomials} 
One central idea in our proof is to define the ``geometric polynomial'' of a set of vectors, and then to identify a sufficient number of its roots. This geometric polynomial is expressed in terms of the folded Wronskian, which we define below. 
\begin{definition}[Folded Wronskian, see \cite{guruswami2016explicit}] Let $f_1(X), \ldots, f_s(X) \in \mathbb{F}_q[X]$ and $\gamma \in \mathbb{F}_q^\times$. We define their $\gamma$-folded Wronskian $W_\gamma\left(f_1, \ldots, f_s\right)(X) \in \left(\mathbb{F}_q[X]\right)^{s\times s}$ by
$$
W_\gamma\left(f_1, \ldots, f_s\right)(X) \stackrel{\text { def }}{=}\left(\begin{array}{ccc}
f_1(X) & \ldots & f_s(X) \\
f_1(\gamma X) & \cdots & f_s(\gamma X) \\
\vdots & \ddots & \vdots \\
f_1\left(\gamma^{s-1} X\right) & \cdots & f_s\left(\gamma^{s-1} X\right)
\end{array}\right).
$$  
\end{definition}

It is well known that the nonsingularity of the folded Wronskian of a set of vectors characterizes the linear independence of these vectors, as stated below.
\begin{lemma}[Folded Wronskian criterion for linear independence, see \cite{guruswami2016explicit,guruswami2013linear,forbes2012identity}]\label{lem:wrons_ind} Let $k<q$ and $\vec{f_1}, \ldots, \vec{f_s} \in \mathbb{F}^k_q$. 
Let $\gamma$ be a generator of $\mathbb{F}_q^\times$.
Then $\vec{f_1}, \ldots, \vec{f_s}$ are linearly independent over $\mathbb{F}_q$ if and only if the folded Wronskian determinant $\operatorname{det} W_\gamma\left(f_1, \ldots, f_s\right)(X) \neq 0$.
\end{lemma}
We are now ready to define the following notion of the ``geometric polynomial'' for a set of vectors.
\begin{definition}[Geometric polynomial based on folded Wronskians]
    Given $L$ non-zero vectors $\vec{f}_1,\vec{f}_2,\dots,\vec{f}_L\in\mathbb{F}_q^k$ such that 
$\dim_{\mathbb{F}_q}\left(\mathrm{Span}_{\mathbb{F}_q}\left\{\vec{f}_1,\vec{f}_2,\dots,\vec{f}_L\right\}\right)=\ell\in[L]$. Then we define the geometric polynomial
$V_{\left\{\vec{f}_i\right\}_{i\in L}}(X)$ as the following monic polynomial
\[\lambda_{i_1,i_2,\dots,i_\ell}\cdot\det W_{\gamma}(f_{i_1},\dots,f_{i_\ell})(X),\]
where $\lambda_{i_1,i_2,\dots,i_\ell}\in\mathbb{F}_q^\times$ 
and $\left\{f_{i_1},\dots,f_{i_\ell}\right\}$ forms a $\mathbb{F}_q$-basis of the space $\mathrm{Span}_{\mathbb{F}_q}\left\{\vec{f}_1,\vec{f}_2,\dots,\vec{f}_L\right\}$.
\end{definition}
We now prove that geometric polynomials are well-defined, based on the key observation that their definition is independent of the choice of the basis $\{f_{i_1},\dots,f_{i_\ell}\}$. 
\begin{lemma}\label{lem:geometric_poly}
  Geometric polynomials are well-defined.  
\end{lemma}
\begin{proof}
Let $V=\mathrm{Span}_{\mathbb{F}_q}\left\{\vec{f}_1,\vec{f}_2,\dots,\vec{f}_L\right\}$. By \Cref{lem:wrons_ind}, it suffices to prove that for any two basis $\{\vec{u}_i\}_{i\in[\ell]},\{\vec{u}'_i\}_{i\in[\ell]}$ of $V$, we have $\det W_{\gamma}(u_1,\dots,u_{\ell})=\lambda\det W_{\gamma}(u'_1,\dots,u'_{\ell})$ for some  $\lambda\in\mathbb{F}^\times_q$. We observe that there must exist a non-singular matrix $A\in\mathbb{F}^{\ell\times\ell}_q$ satisfying $W_{\gamma}(u_1,\dots,u_{\ell})=W_{\gamma}(u'_1,\dots,u'_{\ell})A$. Therefore, $\det W_{\gamma}(u_1,\dots,u_{\ell})=\det A\cdot\det W_{\gamma}(u'_1,\dots,u'_{\ell})$ where $\det A\in\mathbb{F}^\times_q$.
\end{proof}

\subsubsection{Geometric agreement hypergraph provides zeros of a geometric polynomial with multiplicity} 
In this section, we need a relationship between geometric agreement hypergraphs and geometric polynomials. Specifically, this relationship  states that a geometric agreement hypergraph will provide many roots, counted with multiplicity, of the corresponding geometric polynomial, which is essentially implied by \cite[Theorem 14]{guruswami2016explicit}. The number of these roots is closely related to the affine dimensions of its hyperedges. A detailed version of this theorem is stated below. 
\begin{theorem}[{Alternatively stated in \cite[Theorem 14]{guruswami2016explicit}}]\label{thm:lower_bound_root}
Given $L$ distinct non-zero polynomials $f_1,\dots,f_L\in\mathbb{F}^k_q$ with degree at most $k-1$. Let $(\mathcal{V},\mathcal{E})$ be a geometric agreement hypergraph over $\mathcal{V}=\left\{0,\vec{f}_1,\dots,\vec{f}_L\right\}$ where $\mathcal{E}=\big\{e_1,\dots,e_n\subseteq \mathcal{V}\big\}$, it follows that $P(X)=V_{\{f_i\}_{i\in L}}(X)$ has at least $(s-\ell+1)\sum_{i=1}^n\widetilde{\dim}_{\mathbb{F}_q}\left(e_i\right)$ roots, counted with multiplicity, where $\ell=\dim\left(\mathrm{Span}_{\mathbb{F}_q}\left\{\vec{f}_1,\vec{f}_2,\dots,\vec{f}_L\right\}\right)$.
\end{theorem}
\begin{proof}
We want to prove that for any $e_i$, each of $\alpha_i,\dots,\gamma^{s-\ell}\alpha_i$ is a root of $P(X)$ with multiplicity at least $\widetilde{\dim}_{\mathbb{F}_q}\left(e_i\right)$. If it holds for each of $e_i,i\in[n]$, since the sequence  $(\alpha_1,\dots,\alpha_n)$ is \emph{appropriate}, we get the above lower bound. Now let's focus on a fixed $e_i$.

If $\widetilde{\dim}_{\mathbb{F}_q}\left(e_i\right)=t$, then there must exist $\vec{g}_0,\dots,\vec{g}_{t}\in e_i$ such that $\left\{\vec{h}_1=\vec{g}_1-\vec{g}_0,\dots,\vec{h}_t=\vec{g}_t-\vec{g}_0\right\}$ are linear independent. Let $V=\Span_{\mathbb{F}_q}\left\{\vec{f}_1,\dots,\vec{f}_L\right\}$, we can arbitrarily extend $\left\{\vec{h}_u\right\}_{u\in [t]}$ to a basis $\left\{\vec{h}_u\right\}_{u\in[\ell]}$ of $V$.

By \Cref{lem:geometric_poly}, $\det W_{\gamma}(h_1,\dots,h_{\ell})(X)=\lambda V_{\left\{\vec{f}_1,\dots,\vec{f}_L\right\}}(X)$ for some $\lambda\in\mathbb{F}^\times_q$. Therefore it suffices to show that $\gamma^j\alpha_i$ is a root of $\det W_{\gamma}(h_1,\dots,h_{\ell})(X)$ with multiplicity at least $t$ for any $0\leq j\leq s-\ell$. By \Cref{def:agr_graph}, for any $v\in[t]$ and $0\leq j\leq s-1$, we have $h_v(\gamma^j\alpha_i)=g_v(\gamma^j\alpha_i)-g_0(\gamma^j\alpha_i)=0$. It implies for any $v\in[t],j\in\{0\}\cup[s-\ell],u\in[\ell]$, $\gamma^j\alpha_i$ is a root of $h_v(\gamma^{u-1}X)$. Since the $(u,v)$-entry of $ W_{\gamma}(h_1,\dots,h_{\ell})(X)$ is $h_v(\gamma^{u-1}X)$, we know that $(X-\gamma^j\alpha_i)$ is a factor for each entry in the first $t$ columns of $W_{\gamma}(h_1,\dots,h_{\ell})(X)$. Consider computing $\det W_{\gamma}(h_1,\dots,h_{\ell})(X)$ by expanding the matrix along the first $t$ columns, we can see that $\gamma^j\alpha_i$ is a root of $\det W_{\gamma}(h_1,\dots,h_{\ell})(X)$ with multiplicity at least $t$ for any $0\leq j\leq s-\ell$.
\end{proof}
\begin{remark}
We remark that \cref{thm:lower_bound_root} was essentially implied by {\cite[Theorem 14]{guruswami2016explicit}}, although it was originally presented in a different context, specifically for constructing a pseudorandom object called ``subspace designs'' \cite{guruswami2013list, guruswami2016explicit}.  In \cref{append:generalize}, we extend our results by explicitly connecting the notion of subspace designs with codes that achieve list-decoding capacity. See \cref{append:generalize} for more details.
\end{remark}
\subsection{The $\mathrm{Loss}$ function and its upper bound}\label{sec:loss}
From \Cref{thm:lower_bound_root}, we know that a hyperedge $e$ in the geometric agreement hypergraph can provide $(s-\ell+1)\widetilde{\dim}_{\mathbb{F}_q}\left(e\right)$ roots of the geometric polynomial, counted with multiplicity. Our goal is to identify a sufficient number of roots---more than the presumed degree---of a low degree (non-zero) geometric polynomial, which will be instrumental in our final proof by contradiction (see \cref{lem:non_distinct}). To achieve this, we analyze the polynomials in the geometric agreement hypergraph from a geometric perspective. This section represents the novel and crucial part that distinguishes our approach from previous work, allowing us to achieve an optimal list size.

Ideally, we would have $\widetilde{\dim}_{\mathbb{F}_q}\left(e\right)=|e|-1$, which occurs when then vectors in $e$ span an affine subspace of maximal dimension. However, in degenerated cases, we may have $\widetilde{\dim}_{\mathbb{F}_q}\left(e\right)<|e|-1$, resulting in a loss of $(s-\ell+1)\left(|e|-1-\widetilde{\dim}_{\mathbb{F}_q}\left(e\right)\right)$ roots of the geometric polynomial compared with the ideal case. This motivates us to define a loss function below. By establishing an upper bound on the loss function, we can show that the number of roots lost remains limited.
\begin{definition}[Loss function]
We define the loss function $\mathrm{Loss}: \mathcal{E}\to\mathbb{N}$ that sends a hyperedge $e\in\mathcal{E}$ to
    \[\mathrm{Loss}(e):=\max\left(0,|e|-1-\widetilde{\dim}_{\mathbb{F}_q}\left(e\right)\right).\]
\end{definition}
\begin{remark}
Since $\widetilde{\dim}_{\mathbb{F}_q}\left(e\right)$ is related only to the geometric positions of vectors in $e$, the loss function $\mathrm{Loss}(e)$ encodes purely ``geometric'' information about $e$.
\end{remark}
\paragraph{Upper bound on the loss function.} The following upper bound on the loss function plays an crucial role in our proof.
\begin{theorem}\label{thm:bound_loss}
Let $\left\{\vec{f}_i\right\}_{i\in[L]}$ be a set of distinct non-zero vectors in $\mathbb{F}_q^k$ and vertices $\mathcal{V}:=\left\{0,\vec{f}_1,\vec{f}_2,\dots,\vec{f}_L\right\}$. Let $\dim_{\mathbb{F}_q}\left(\mathrm{Span}_{\mathbb{F}_q}\left\{\vec{f}_1,\vec{f}_2,\dots,\vec{f}_L\right\}\right)=\ell\in[L]$. Consider a geometric agreement  hypergraph $(\mathcal{V},\mathcal{E})$ with $n$ hyperedges $\mathcal{E}=\{e_1,e_2,\dots,e_n\subseteq\mathcal{V}\}$ such that for any proper subset $\mathcal{H}\subsetneq \mathcal{V}$ with $|\mathcal{H}|\ge 2$, we have $\mathrm{wt}(\mathcal{H},\mathcal{E}|_{\mathcal{H}})<\frac{(|\mathcal{H}|-1)k}{s-|\mathcal{H}|+2}$.
Then, we have the following upper bound on the loss function:
\[\sum_{i\in[n]}\mathrm{Loss}(e_i)\leq\frac{(L-\ell)k}{s-L+1}.\]    
\end{theorem}
\begin{remark}
    Before proving the above theorem, we first provide a warm-up example below, which may provide some useful insights to our final proof. Intuitively, when all hyperedges contain the vertex $0$, it is straightforward to identify a subset (denoted as $\mathcal{H}$ below) consisting of $L+1-\ell$ vertices, such that the total sum of losses is bounded by $\mathrm{wt}\left(\mathcal{H}, \mathcal{E}|_{\mathcal{H}}\right)$. Consequently, by applying the upper bound assumption on $\mathrm{wt}\left(\mathcal{H}, \mathcal{E}|_{\mathcal{H}}\right)$, we obtain the desired bound.
\end{remark}
\begin{proof}[{ \bf A warm-up case: when all the hyperedges contain $\{0\}$}] As a warm-up, let's assume for any hyperedge $e_i$ where $i\in[n]$, we have $0\in e_i$. We show how to prove \Cref{thm:bound_loss} under this restriction.

Without loss of generality, we can assume $\vec{f}_1,\dots,\vec{f}_{\ell}$ are linearly independent. Let $\mathcal{T}=\left\{\vec{f}_1,\dots,\vec{f}_{\ell}\right\}$ and $\mathcal{H}=\{0,\vec{f}_1,\dots,\vec{f}_{L}\}\backslash\mathcal{T}$. Fixed any $i\in[n]$, we classify the elements in $e_i$ as $e_i=\left\{\vec{f}_{u_1},\dots,\vec{f}_{u_{a}}\right\}\cup\left\{\vec{f}_{v_1},\dots,\vec{f}_{v_{b}}\right\}$ where $u_1,\dots,u_a\in\mathcal{T}$ and $v_1,\dots,v_b\in\mathcal{H}$. From \Cref{fact:transform_tildedim}, since $0\in e_i$, we know that
\[\widetilde{\dim}_{\mathbb{F}_q}\left(e_i\right)=\dim_{\mathbb{F}_q}\operatorname{Span}_{\mathbb{F}_q}\left\{\vec{f}_{u_1},\dots,\vec{f}_{u_a},\vec{f}_{v_1},\dots,\vec{f}_{v_b}\right\}\ge a=\Big|{e_i}|_{\mathcal{T}}\Big|.\]
Therefore, we have
\[\begin{split}
    \sum_{i\in[n]}\mathrm{Loss}(e_i)&\leq \sum_{i\in[n]}\max\left(0,|e_i|-1-\Big|{e_i}|_{\mathcal{T}}\Big|\right)\\
    &=\sum_{i\in[n]}\max\left(0,\Big|{e_i}|_{\mathcal{H}}\Big|-1\right)\leq \mathrm{wt}\left(\mathcal{H},\mathcal{E}|_{\mathcal{H}}\right)\\
&\leq \frac{\left(|\mathcal{H}|-1
\right)k}{s-(|\mathcal{H}|-1)+1}\leq\frac{(L-\ell)k}{s-L+1}.\qedhere
\end{split}\]
\end{proof}
\paragraph{A linear-algebraic lemma.} To extend the above warm-up case, we will introduce a linear-algebraic lemma and prove \Cref{thm:bound_loss} in full generality.
Before that, we provide a necessary definition below, which is a standard notion in matroid theory.
\begin{definition}[Flat]
Given $m$ vectors $\mathcal{F}=\left\{\vec{f}_1,\dots,\vec{f}_m\right\}\in\mathbb{F}^k_q$, a \emph{flat} of $\mathcal{F}$ with dimension $\ell$ is a proper subset $\mathcal{H}\subsetneq\mathcal{F}$ with $\widetilde{\dim}_{\mathbb{F}_q}\left(\mathcal{H}\right)=\ell$ such that  $\widetilde{\dim}_{\mathbb{F}_q}\left(\left\{\vec{f}\right\}\cup\mathcal{H}\right)>\widetilde{\dim}_{\mathbb{F}_q}\left(\mathcal{H}\right)$ for any $\vec{f}\in\mathcal{F}\backslash\mathcal{H}$.
\end{definition}
Now we are ready to introduce a simple but very important linear-algebraic lemma below, which will be crucially applied into the proof of \cref{thm:bound_loss}.
\begin{lemma}\label{lem:dg}
      Given $\vec{f}_1,\vec{f}_2,\dots,\vec{f}_L\in\F_q^k$ such that $\vec{f}_i\neq \vec{f}_j \neq 0$ for any $i\neq j\in[L].$ Let the dimension of $\mathrm{Span}_{\mathbb{F}_q}\left\{\vec{f}_1,\vec{f}_2,\dots,\vec{f}_L\right\}$ equals to $\ell\in[L],$ then there exists a partition $\{\mathcal{H}_i\}_{i\in[\ell+1]}$ of $\left\{0,\vec{f}_1,\vec{f}_2,\dots,\vec{f}_L\right\}$ and for any ${\left\{\vec{h}_i\in \mathcal{H}_i\right\}_{i\in[\ell+1]}}$ we have 
    \[\widetilde{\dim}_{\mathbb{F}_q}\left(\vec{h}_1,\vec{h}_2,\dots,\vec{h}_{\ell+1}\right)=\ell.\]
\end{lemma}
\begin{proof}
    We prove this theorem by induction on $\ell\ge1.$ 
    
    When $\ell=1$, we can consider the partition $\mathcal{H}_1=\{0\},\mathcal{H}_2=\left\{\vec{f}_1,\dots,\vec{f}_L\right\}$. This partition satisfies the desired condition.

    When $\ell>1$, without loss of generality we can assume $\vec{f}_1,\dots,\vec{f}_{\ell}$ are linear independent over $\mathbb{F}_q$. Let $V=\mathrm{Span}_{\mathbb{F}_q}\left\{\vec{f}_1,\dots,\vec{f}_{\ell-1}\right\}$ and $\mathcal{F}=\{0\}\cup\left\{\vec{f}_i:\vec{f}_i\in V\text{ for }i\in[L]\right\}$, by induction there exists a partition $\{\mathcal{H}_i\}_{i\in[\ell]}$ of $\mathcal{F}$ such that for any $\left\{\vec{h}_i\in\mathcal{H}_i\right\}_{i\in[\ell]}$, we have $\widetilde{\dim}_{\mathbb{F}_q}\left(\vec{h}_1,\dots,\vec{h}_{\ell}\right)=\ell-1$. Define $\mathcal{V}:=\left\{0,\vec{f}_1,\dots,\vec{f}_L\right\}$ and $\mathcal{H}_{\ell+1}:=\mathcal{V}\backslash\mathcal{F}$. Since $\vec{f}_\ell\in\mathcal{H}_{\ell+1}$, $\mathcal{H}_{\ell+1}$ must be a non-empty set, so $\{\mathcal{H}_i\}_{i\in[\ell+1]}$ is a partition of $\mathcal{V}$. Moreover, given any choices of $\left\{\vec{h}_i\in\mathcal{H}_i\right\}_{i\in[\ell+1]}$, it follows that $\ell-1=\widetilde{\dim}_{\mathbb{F}_q}\left(\vec{h}_1,\dots,\vec{h}_{\ell}\right)<\widetilde{\dim}_{\mathbb{F}_q}\left(\vec{h}_1,\dots,\vec{h}_{\ell+1}\right)\leq \widetilde{\dim}_{\mathbb{F}_q}\left(\vec{h}_1,\dots,\vec{h}_{\ell}\right)+1$ since $\mathcal{F}$ is a flat of $\mathcal{V}$ and $\vec{h}_{\ell+1}\notin \mathcal{F}$. Therefore, $\widetilde{\dim}_{\mathbb{F}_q}\left(\vec{h}_1,\dots,\vec{h}_{\ell+1}\right)=\ell$, we get a desired partition.
\end{proof}
Based on the aforementioned lemma, we provide the complete proof of our main theorem in this section (\cref{thm:bound_loss}) below.
\begin{proof}[{\bf Proof of \cref{thm:bound_loss}}]
     By \cref{lem:dg}, we know that there exists a partition $\{\mathcal{H}_i\}_{i\in[\ell+1]}$ of $\left\{0,\vec{f}_1,\vec{f}_2,\dots,\vec{f}_L\right\}$ and for any ${\left\{\vec{h}_i\in \mathcal{H}_i\right\}_{i\in[\ell+1]}}$ we have 
    $\widetilde{\dim}_{\mathbb{F}_q}\left(\vec{h}_1,\vec{h}_2,\dots,\vec{h}_{\ell+1}\right)=\ell$. Then, for any $r\in[\ell+1]$ and $\{i_1,i_2,\dots,i_r\}\subseteq[\ell+1]$, since $\dim_{\mathbb{F}_q}\left(\mathrm{Span}_{\mathbb{F}_q}\left\{\vec{f}_1,\vec{f}_2,\dots,\vec{f}_L\right\}\right)=\ell$, we have $\widetilde{\dim}_{\mathbb{F}_q}\left(\vec{h}_{i_1},\vec{h}_{i_2},\dots,\vec{h}_{i_{r}}\right)=r-1.$ Then, for any hyperedge $e\subseteq\mathcal{V}$, we have 
    \[\widetilde{\dim}_{\mathbb{F}_q}(e)\ge \left(\sum_{i\in[\ell+1]}\chi_{e}(\mathcal{H}_i)\right)-1,\]
    where $\chi_{e}(\mathcal{H}):=\left\{
\begin{aligned}
1 &  & \text{if }e\cap \mathcal{H}\neq\emptyset \\
0 &  & \text{otherwise}
\end{aligned}
\right.$ for any $\mathcal{H}\subseteq\mathcal{V}$,
which implies 
\[\mathrm{Loss}(e)\leq |e|-1-\widetilde{\dim}_{\mathbb{F}_q}(e)\leq |e|-\sum_{i\in[\ell+1]}\chi_{e}(\mathcal{H}_i)\leq \sum_{i\in[\ell+1]}\max\bigg(0,\Big|{e_i}|_{\mathcal{H}_i}\Big|-1\bigg).\]
Therefore, we have
\[\begin{split}
    \sum_{i\in[n]}\mathrm{Loss}(e_i)&\leq \sum_{i\in[n]}\sum_{j\in[\ell+1]}\max\bigg(0,\Big|{e_i}|_{\mathcal{H}_j}\Big|-1\bigg)\\
    &=\sum_{j\in[\ell+1]}\mathrm{wt}\left(\mathcal{H}_j,\mathcal{E}|_{\mathcal{H}_j}\right)\leq \sum_{j\in[\ell+1]}\frac{\left(|\mathcal{H}_j|-1\right)k}{s-|\mathcal{H}_j|+2}\\
&\leq \frac{\bigg(\sum_{j\in[\ell+1]}{\big(|\mathcal{H}_j|-1\big)}\bigg)k}{s-L+1}=\frac{(L-\ell)k}{s-L+1},
\end{split}\]
which completes the proof of \cref{thm:bound_loss}.
\end{proof}
\subsection{Putting it together}\label{sec:puttogether}
 By combining all the aforementioned results, we provide a technical lemma below, which will serve as a cornerstone in the proof of \cref{thm:main}.
\begin{lemma}\label{lem:non_distinct}
Consider any $m\ge 2$ vectors $\vec{f}_1,\dots,\vec{f}_m\in\mathbb{F}^k_q$ and a geometric agreement hypergraph $(\mathcal{V},\mathcal{E})$ over $\mathcal{V}=\left\{\vec{f}_1,\dots,\vec{f}_m\right\}$. If $\mathrm{wt}(\mathcal{V},\mathcal{E})\ge \frac{(m-1)k}{s-m+2}$, then $\vec{f}_1,\dots,\vec{f}_m$ cannot be distinct.
\end{lemma}
\begin{proof}
Since $\mathrm{wt}(\mathcal{V},\mathcal{E})\ge \frac{(|\mathcal{V}|-1)k}{s-|\mathcal{V}|+2}$, there must exist a minimal subset $\mathcal{V}_0\subseteq \mathcal{V}$ with $|\mathcal{V}_0
|\ge 2$ and 
$\mathcal{V}_0$ satisfies the following conditions. 
\begin{itemize}
\item $\mathrm{wt}(\mathcal{V}_0,\mathcal{E}|_{\mathcal{V}_0})\ge\frac{(|\mathcal{V}_0|-1)k}{s-|\mathcal{V}_0|+2}$
\item For any proper subset $\mathcal{H}\subsetneq \mathcal{V}_0$ with $|\mathcal{H}|\ge 2$, $\mathrm{wt}(\mathcal{H},\mathcal{E}|_{\mathcal{H}})<\frac{(|\mathcal{H}|-1)k}{s-|\mathcal{H}|+2}$.
\end{itemize}
Let $m':=|\mathcal{V}_0|\ge 2$ and $\mathcal{V}_0=\{\vec{f}_{i_1},\dots,\vec{f}_{i_{m'}}\}$. Suppose by contradiction that $\vec{f}_1,\dots,\vec{f}_m$ are distinct, then $\vec{g}_1:=\vec{f}_{i_1}-\vec{f}_{i_1},\dots,\vec{g}_{m'}:=\vec{f}_{i_{m'}}-\vec{f}_{i_1}$ are also distinct. Moreover, by the definition of $\mathcal{V}_0$, there must exist a corresponding  geometric agreement hypergraph $(\mathcal{V}',\mathcal{E}')$ such that
\begin{itemize}
\item $\mathcal{V}'=\{\vec{g}_1=0,\vec{g}_2,\dots,\vec{g}_{m'}\},\mathcal{E}'=\{e_1,\dots,e_n\subseteq \mathcal{V}' \}$
\item $\mathrm{wt}(\mathcal{V}',\mathcal{E}')\ge\frac{(|\mathcal{V}'|-1)k}{s-|\mathcal{V}'|+2}$
\item For any proper subset $\mathcal{H}\subsetneq \mathcal{V}'$ with $|\mathcal{H}|\ge 2$, $\mathrm{wt}(\mathcal{H},\mathcal{E}'|_{\mathcal{H}})<\frac{(|\mathcal{H}|-1)k}{s-|\mathcal{H}|+2}$.
\end{itemize}
Let $P(X):=V_{\{\vec{g}_{2},\dots,\vec{g}_{m'}\}}(X)$. Since $\vec{g}_1=0,\vec{g}_2,\dots,\vec{g}_{m'}$ are distinct, we know $P(X)$ is a non-zero 
 polynomial with degree at most $\ell(k-1)$ where $\ell=\dim_{\mathbb{F}_q}\left(\Span_{\mathbb{F}_q}\{\vec{g}_2,\dots,\vec{g}_{m'}\}\right), 1\leq \ell\leq m'-1$. By \Cref{thm:lower_bound_root}, we know $P(X)$ has at least $(s-\ell+1)\sum_{i=1}^n\widetilde{\dim}_{\mathbb{F}_q}\left(e_i\right)\ge (s-\ell+1)\bigg(\mathrm{wt}(\mathcal{V}',\mathcal{E}')-\sum^n_{i=1}\mathrm{Loss}(e_i)\bigg)$ roots counting multiplicity. Moreover, by the weight lower bound and \Cref{thm:bound_loss}, we have:
\begin{equation*}
(s-\ell+1)\left(\mathrm{wt}(\mathcal{V}',\mathcal{E}')-\sum^n_{i=1}\mathrm{Loss}(e_i)\right)
\ge(s-\ell+1)\left(\frac{(m'-1)k}{s-m'+2}-\frac{(m'-1-\ell)k}{s-m'+2}\right)
\ge\ell k
\end{equation*}
However, the degree of $P(X)$ is at most $\ell(k-1)<\ell k$. Therefore, $P(X)=0$, which is a contradiction. We conclude that $\vec{f}_1,\dots,\vec{f}_m$ cannot be distinct.
\end{proof}
We are now ready to prove our main results on the list-decodability of folded Reed--Solomon codes, \cref{thm:main} and \cref{cor:main}.
\begin{theorem}[Restatement of \cref{thm:main}]\label{thm:main_restate}
For any $L\ge1, s,n,k\in\mathbb{N}^+,q>n,s\ge L$ and generator $\gamma$ of $\mathbb{F}_q^\times$, the folded RS code $\mathsf{FRS}_{n,k}^{(s,\gamma)}(\alpha_1,\alpha_2,\dots,\alpha_n)$ with appropriate evaluation points in $\mathbb{F}_q$ is $\left(\frac{L}{L+1}\left(1-\frac{sR}{s-L+1}\right),L\right)$ list-decodable.
\end{theorem}
\begin{proof}
Suppose by contradiction that there exists $L+1$ distinct polynomials $f_1,\dots,f_{L+1}\in\mathbb{F}_q[x]_{<k}$ and a received word $y\in(\mathbb{F}^s_q)^n$ such that for each $i\in[L+1]$, the codeword $c_i:=\mathcal{C}(f_i)$ of $f_i$ has relative hamming distance at most $\frac{L}{L+1}\left(1-\frac{sR}{s-L+1}\right)$ from $y$. Then for any $i\in[L+1]$, there is a subset $I_i\subseteq [n]$ where $|I_i|\ge \frac{n}{L+1}+\frac{Lk}{(L+1)(s-L+1)}$, such that for each $t\in I_i$, $c_{i}[t]=y[t]$. For each $j\in[n]$, we define $e_j=\left\{\vec{f}_i\colon j\in I_i\right\}$ as the set of polynomials whose codewords match with $y$ on position $j$ in this bad list. Let $(\mathcal{V},\mathcal{E})$ denote the corresponding geometric agreement hypergraph, there is $\mathcal{V}=\left\{\vec{f}_1,\dots,\vec{f}_{L+1}\right\}$, and $\mathcal{E}=(e_1,\dots,e_n)$.
The weight of this hypergraph can be bounded by
\begin{equation*}
\mathrm{wt}(\mathcal{V},\mathcal{E})\ge \sum_{j=1}^n\bigg(|e_j|-1\bigg)\ge\left(\sum^{L+1}_{i=1}|I_i|\right)-n\ge\frac{Lk}{s-L+1}.
\end{equation*}
Then, by \Cref{lem:non_distinct}, $f_1,\dots,f_{L+1}$ cannot be distinct, which is a contradiction. We conclude that such a bad list $f_1,\dots,f_{L+1}$ doesn't exist.
\end{proof}
Finally, we conclude this section by restating and proving (a detailed version of) \cref{cor:main} as follows.
\begin{corollary}[Restatement of \cref{cor:main}]\label{cor:exact_fold}
    For any $\epsilon>0$, $N>k\ge1,$ $L>\frac{1-R-\epsilon}{\epsilon}$,  $s>\frac{L(L-1)R}{\epsilon L-(1-R-\epsilon)}+L-1$ where $s|N$, generator $\gamma$ of $\mathbb{F}^{\times}_q$ and $R=\frac{k}{N}$. Let 
$n=\frac{N}{s}$ and  $\alpha_1,\alpha_2,\dots,\alpha_n\in\F_q$ be appropriate. Then the folded RS code $\mathsf{FRS}_{n,k}^{(s,\gamma)}(\alpha_1,\alpha_2,\dots,\alpha_n)$ over the alphabet $\mathbb{F}_q^s$ is $\big(1-R-\epsilon,L\big)$ list-decodable.
Therefore, by choosing different 
folding parameter $s$, the folded RS code $\mathsf{FRS}_{n,k}^{(s,\gamma)}(\alpha_1,\alpha_2,\dots,\alpha_n)$ is
\begin{itemize}
\item[$(a)$] $(1-R-\epsilon,\lfloor\frac{1-R}{\epsilon}\rfloor)$ list-decodable, when $s>\widetilde{s}$, where $\widetilde{s}$ is some constant only depending on $R,\epsilon$.
\item[$(b)$] $(1-R-\epsilon,\lceil\frac{1-R}{\epsilon}\rceil)$ list-decodable, when $s=\Theta({1}/{\epsilon^3})$.
\item[$(c)$] $(1-R-\epsilon,\lceil\frac{1}{\epsilon}\rceil)$ list-decodable, when $s=\Theta({1}/{\epsilon^2})$.
\end{itemize}
\end{corollary}
\begin{proof}
    For any $L>\frac{1-R-\epsilon}{\epsilon}$ and $s>\frac{L(L-1)R}{\epsilon L-(1-R-\epsilon)}+L-1$, we have $s\ge L$. From \Cref{thm:main}, we know $\mathsf{FRS}_{n,k}^{(s,\gamma)}(\alpha_1,\alpha_2,\dots,\alpha_n)$ is a $\left(\rho:=\frac{L}{L+1}\left(1-\frac{sR}{s-L+1}\right), L\right)$ list-decodable code. We calculate that
    \begin{equation*}
    \rho=\frac{L}{L+1}\left(1-\frac{sR}{s-L+1}\right)
    >\frac{L}{L+1}\left(1-R-\frac{\epsilon L-(1-R-\epsilon)}{L}\right)
    =1-R-\epsilon
    \end{equation*}
    Therefore, $\mathsf{FRS}_{n,k}^{(s,\gamma)}(\alpha_1,\alpha_2,\dots,\alpha_n)$ is $\left(1-R-\epsilon, L\right)$ list-decodable.

Now we choose different parameter settings for $L,s$ that satisfy the above restrictions and inspect list-decodability of $\mathsf{FRS}_{n,k}^{(s,\gamma)}\left(\alpha_1,\alpha_2,\dots,\alpha_n\right)$, stated below
\begin{itemize}
\item[$(a)$] When $L=\lfloor\frac{1-R}{\epsilon}\rfloor, s>\widetilde{s}=\frac{L(L-1)R}{\epsilon L-(1-R-\epsilon)}+L-1$, it is $\left(1-R-\epsilon,\lfloor\frac{1-R}{\epsilon}\rfloor\right)$ list-decodable.
\item[$(b)$] When $L=\lceil\frac{1-R}{\epsilon}\rceil,s>\lceil\frac{3}{\epsilon^3}\rceil$, it is $\left(1-R-\epsilon,\lceil\frac{1-R}{\epsilon}\rceil\right)$ list-decodable.
\item[$(c)$] When $L=\lceil\frac{1}{\epsilon}\rceil, s>\lceil\frac{3}{\epsilon^2}\rceil$, it is $\left(1-R-\epsilon,\lceil\frac{1}{\epsilon}\rceil\right)$ list-decodable.\qedhere
\end{itemize}
\end{proof}
\section{Improved Upper Bound on the Radius for List-Recoverability of Folded RS Codes}\label{sec:list_recover}
In this section, we derive an upper bound on the radius $\rho$, beyond which folded RS codes fail to be $(\rho, \ell, L)$ list-recoverable. Although this bound may not be tight for all parameter settings $(\ell, L)$, we conjecture that it is nearly tight when $L + 1$ is a power of $\ell$. In particular, we demonstrate that for $\ell = 2$ and $L = 3$, this bound is tight. Our result significantly improves upon the previous bound from \cite{gst22b} and rules out the possibility that folded RS codes could achieve list-recovery capacity.
Surprisingly, since these codes achieve list-decoding capacity, as shown in \cref{cor:main}, this result highlights an intrinsic separation between list-decodability and list-recoverability.

Given an alphabet $\Sigma$, a product set $S=S_1\times \cdots\times S_n\in(2^\Sigma)^n$, and a (corrupted) codeword $c\in \Sigma^n$, we denote by $\mathrm{dist}(c,S)$ the number of indices $i\in[n]$ such that $c_i\notin S_i$. 
We also define the agreement as $\mathrm{Agr}(c,S):=n-\mathrm{dist}(c,S)$, which will be used throughout this section. All results in this section can also be adapted to univariate multiplicity codes. For brevity, we will focus on folded Reed--Solomon codes here. 

The following theorem provides our major upper bound on the list-recovery radius of folded RS codes.
\begin{theorem}\label{thm:listrec}
Let $s\ge 1$, $2\leq \ell\leq L,q\ge\ell$, $\gamma$ be a  generator of $\mathbb{F}^{\times}_q$, and $m=\lceil\log_{\ell}(L+1)\rceil>1$. Suppose $\frac{k-1}{sn}\leq \frac{m-1}{m}$ and $\frac{k-1}{(m-1)s}\ge 1$. If a folded Reed--Solomon code $\mathsf{FRS}^{(s,\gamma)}_{n,k}(\alpha_1,\alpha_2,\dots,\alpha_n)$ with appropriate evaluation points in $\mathbb{F}_q$ is  $(\rho,\ell,L)$ list-recoverable, then 
    \[\rho\leq\frac{L+1-\ell}{L+1}\left(1-\left\lfloor\frac{m}{m-1}\left\lfloor\frac{k-1}{s}\right\rfloor\right\rfloor\frac{1}{n}\right).\]
\end{theorem}

Before proving \cref{thm:listrec}, we first derive a simplified form for ease of use.

\begin{corollary}[Simplified form of \cref{thm:listrec}. Restatement of \cref{thm:bound_clean}]\label{cor:bound_clean_app}
Let $s\ge 1$, $2\leq \ell\leq L,q\ge\ell$, generator $\gamma$ of $\mathbb{F}^{\times}_q$ and $m=\lceil\log_{\ell}(L+1)\rceil>1$. Suppose $R=\frac{k}{sn}\leq \frac{m-1}{m}$ and $\frac{k-1}{s}\ge m$. If a folded Reed--Solomon code $\mathsf{FRS}^{(s,\gamma)}_{n,k}(\alpha_1,\alpha_2,\dots,\alpha_n)$ of rate $R$ with appropriate evaluation points in $\mathbb{F}_q$ is  $(\rho,\ell,L)$ list-recoverable, then 
    \[\rho\leq\frac{L+1-\ell}{L+1}\left(1-\frac{mR}{m-1}\right)+\frac{5}{n}.\]  
\end{corollary}
\begin{proof}
By \cref{thm:listrec}, since $m\ge 2$, we have $\frac{m}{m-1}\leq 2$ and 
\begin{align*}
\rho&\leq\frac{L+1-\ell}{L+1}\left(1-\left\lfloor\frac{m}{m-1}\left\lfloor\frac{k-1}{s}\right\rfloor\right\rfloor\frac{1}{n}\right)
\\
&\leq \frac{L+1-\ell}{L+1}\left(1-\left(\left\lfloor\frac{m(k-1)}{(m-1)s}\right\rfloor-2\right)\frac{1}{n}\right)\\
&\leq \frac{L+1-\ell}{L+1}\left(1-\left(\left\lfloor\frac{mk}{(m-1)s}\right\rfloor-4\right)\frac{1}{n}\right)\\
&\leq \frac{L+1-\ell}{L+1}\left(1-\left(\frac{mk}{(m-1)s}-5\right)\frac{1}{n}\right)\\
&\leq \frac{L+1-\ell}{L+1}\left(1-\frac{mR}{m-1}\right)+\frac{5}{n}. \qedhere
\end{align*}
\end{proof}

As a corollary, we obtain \cref{cor:rewrite_bound} in an alternative form, demonstrating the impossibility of removing the exponential dependency on $\frac{1}{\epsilon}$ in the  $(1-R-\epsilon, \ell, (\frac{\ell}{\epsilon})^{O(\frac{1+\log{\ell}}{\epsilon})})$ list-recoverability of FRS codes established in \cite{kopparty2018improved, kopparty2023improved, tamo2024tighter}, let alone matching the list-recovery capacity $(1-R-\epsilon, \ell, O(\frac{\ell}{\epsilon}))$ achieved by random codes over large alphabets (See the discussion about the bound \eqref{recovercapacity}). Since we have shown in \cref{cor:main} that these codes achieve list-decoding capacity, this implies a separation between list-decodability and list-recoverability.
\begin{corollary}[Restatement of \cref{cor:rewrite_bound}]\label{cor:FRSnot}
For any constants $0<R<1, 2\leq \ell\leq q, s\ge 1, 0<\epsilon<\frac{R(1-R)}{4}$, and generator $\gamma$ of $\mathbb{F}^{\times}_q$, if $k$ and $n$ are sufficiently large, then any rate $R=\frac{k}{sn}$ folded Reed--Solomon code $\mathsf{FRS}^{(s,\gamma)}_{n,k}(\alpha_1,\dots,\alpha_n)$ with appropriate evaluation points in $\mathbb{F}_q$ cannot be $(1-R-\epsilon, \ell, \ell^{\frac{R}{2\epsilon}-1}-1)$ list-recoverable.
\end{corollary}
\begin{proof}
When $\epsilon<\frac{R(1-R)}{4}$, we must have $R-2\epsilon\ge \frac{R}{2}$, and $R\leq 1-\frac{4\epsilon}{R}\leq 1-\frac{2\epsilon}{R-2\epsilon}=\frac{R/(2\epsilon)-2}{R/(2\epsilon)-1}\leq \frac{t-1}{t}$ where $t:=\lceil \frac{R}{2\epsilon}\rceil-1$. Applying \cref{cor:bound_clean_app} with $m=t$, $L=\ell^t-1$, and sufficiently large $n$ shows that the FRS code in the statement cannot be $(\rho, \ell, \ell^t-1)$ list-recoverable, where 
\[
\rho:=\frac{L+1-\ell}{L+1}\left(1-\frac{tR}{t-1}\right)+\frac{5}{n}\leq 1-\frac{tR}{t-1}\leq 1-R-\eps.
\] 
Here the first inequality holds since $n$ is large enough, and the second inequality holds since $t=\lceil \frac{R}{2\epsilon}\rceil-1\leq \frac{R}{\eps}+1$.
It follows that the code is not $(1-R-\epsilon, \ell, \ell^{\frac{R}{2\epsilon}-1}-1)$ list-recoverable.
\end{proof}
Recall the previous best known upper bound from \cite{gst22b} only shows that the codes in \cref{cor:FRSnot} cannot be $(1-R-\epsilon, \ell, \frac{\ell(1-R)}{\epsilon}-2)$ list-recoverable for sufficiently large $n$.
In contrast, our \cref{cor:rewrite_bound} yields an exponential list size of $\ell^{\frac{R}{2\epsilon}-1}-1$, compared with their linear list size of $\frac{\ell(1-R)}{\epsilon}-2$. Thus, our bound represents a significantly tighter constraint.
\begin{remark}
Following the notation in \cref{cor:rewrite_bound}, by setting $s=1$, we conclude that any Reed--Solomon code of rate $R=k/n$ with distinct evaluation points cannot be $(1-R-\epsilon, \ell, \ell^{\frac{R}{2\epsilon}-1}-1)$ list-recoverable when the block length $n$ is sufficiently large.
\end{remark}

\subsection{Proof of \cref{thm:listrec}}

We now prove \cref{thm:listrec} as follows.
Let $p:=\left\lfloor\frac{m}{m-1}\lfloor\frac{k-1}{s}\rfloor\right\rfloor$. Since $\frac{k-1}{sn}\leq\frac{m-1}{m}$ and $\frac{k-1}{(m-1)s}\ge 1$, we know $n\ge p\ge m$. Define $Q_i(X):=\prod_{j=0}^{s-1}(X-\gamma^j\alpha_i)$ for all $i\in[n]$. Let $M(X):=\prod_{i=1}^pQ_i(X)$. For $t\in[m]$, let \[M_t(X):=\prod_{i\in[p],i\equiv t\bmod{m}}Q_i(X).\]
Then, for $t\in[m]$, define $f_t(X):=\frac{M(X)}{M_t(X)}$. See \cref{fig:ex} for an example of this construction. Next, we claim $\deg{f_t}\leq k-1$ for all $t\in[m]$. \cref{fig:ex} provides a ``visual proof'' illustrating why this holds. Formally, let 
$\lfloor\frac{k-1}{s}\rfloor=a(m-1)+r$ where $a=\left\lfloor\left\lfloor\frac{k-1}{s}\right\rfloor\frac{1}{m-1}\right\rfloor$ and $r\in[m-2]$. For any $t\in[m]$, we can bound $\deg f_t$ as follows.
\begin{align*}
\deg{f_t}&\leq s\left(p-\left\lfloor\frac{p}{m}\right\rfloor\right)\\
&\leq s\left(\left\lfloor\frac{m\left(a(m-1)+r\right)}{m-1}\right\rfloor-\left\lfloor\left\lfloor\frac{m\left(a(m-1)+r\right)}{m-1}\right\rfloor\frac{1}{m}\right\rfloor\right)\leq s\left(am+\left\lfloor\frac{rm}{m-1}\right\rfloor-a\right)\\
&=s\left(a(m-1)+r+\left\lfloor\frac{r}{m-1}\right\rfloor\right)=s\left(a(m-1)+r\right)=s\left\lfloor\frac{k-1}{s}\right\rfloor\leq k-1,
\end{align*}
where $\left\lfloor\frac{r}{m-1}\right\rfloor=0$ follows from the fact that $r<m-1$. 

Fix $\ell$ distinct elements $\beta_1,\dots,\beta_{\ell}\in\mathbb{F}_q$, which is possible as $q\geq \ell$. Consider a set $F$ of polynomials over $\mathbb{F}_q$ of degree at most $k-1$, defined as
\[F:=\left\{ \sum_{i=1}^{m}\beta_{j_i}f_i(X)\colon (j_1,\dots,j_{m})\in[\ell]^{m} \right\}. \]
\begin{claim}\label{cl:distinct}
$F$ contains $\ell^m$ distinct elements.
\end{claim}
\begin{proof}
It suffices to prove that $f_1(X),\dots,f_m(X)$ are $\mathbb{F}_q$-linear independent, or equivalently, the equation $\sum_{i=1}^mc_if_i(X)=0$ does not have a non-zero solution $(c_1,\dots,c_m)\in\mathbb{F}_q^m$. Assume to the contrary that there exists a non-zero solution $(c_1,\dots,c_m)$. Then $c_t\neq 0$ for some $t\in [m]$. Fix such $t$. Let $P(X)=\sum_{i=1}^mc_if_i(X)=0$. 
Choose $j$ to be the unique integer in $[m]$ such that $j\equiv t\bmod m$. 
Since $p\ge m$, we have $j\in [p]$. 
By definition, we have $Q_j(X)\mid M_t(X)$ and  $Q_j(X)\mid f_{t'}(X)$ for $t'\neq t$. In particular, as $Q_j(\alpha_j)=0$, we have $M_t(\alpha_j)=0$ and $f_{t'}(\alpha_j)=0$ for $t'\neq t$. 
As $\alpha_1,\dots,\alpha_n$ are appropriate, the polynomials $Q_1,\dots, Q_n$ are mutually coprime.
It follows by definition that $f_t(X)=M(X)/M_t(X)$ does not vanish at $\alpha_j$.
Therefore, $P(\alpha_j)=\sum_{i=1}^m c_i f_i(\alpha_j)=c_t f_t(\alpha_j)\neq 0$. But this contradicts the assumption that $P(X)=0$.
\end{proof}
For any $g\in F$, which is a polynomial over $\F_q$ of degree at most $k-1$, denote by $\mathcal{C}(g)\in (\mathbb{F}_q^s)^n$ the corresponding codeword. Similarly, we use $\mathcal{C}(f_v)\in(\mathbb{F}^s_q)^n$ to denote the codeword corresponding to $f_v$ for each $v\in[m]$. 
By \cref{cl:distinct}, we know $|F|=\ell^m\geq L+1$.
Arbitrarily choose $L+1$ distinct polynomials $g_1(X),\dots,g_{L+1}(X)\in F$.

\begin{figure}[t!]
\centering 
\includegraphics[width=0.8\textwidth]{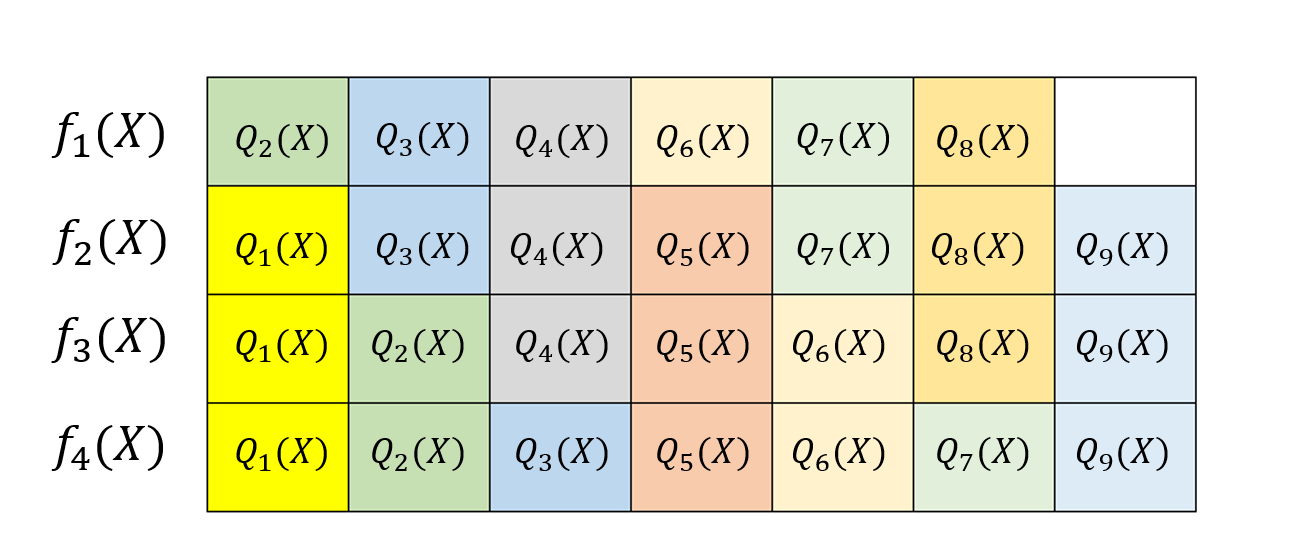} 
\caption{An example of our construction in \cref{thm:listrec} is provided here, with parameters set to $m=4$, $\lfloor\frac{k-1}{s}\rfloor=7$, and $p=9$. For each $i\in[m]$, $f_i(X)$ is a product of distinct factors in the set $\left\{Q_u(X)\right\}_{u\in[p]}$, listed in the $i$-th row of the table. The factors $Q_u$ are arranged in ascending order in $u$. 
Each $Q_u$ appears exactly $m-1$ times in the table by the definition of the polynomials $f_i$.
Since there are $m$ rows and $p=\left\lfloor\frac{m}{m-1}\lfloor\frac{k-1}{s}\rfloor\right\rfloor$ distinct factors $Q_u$, with each factor appearing exactly $m-1$ times, the way we fill these factors in the table 
guarantees that each row contains no more than $\left\lceil\frac{p(m-1)}{m}\right\rceil \leq\lfloor\frac{k-1}{s}\rfloor$ factors $Q_u$. 
As each $Q_u$ has degree $s$, it follows that $\deg{f_i}\leq s\lfloor\frac{k-1}{s}\rfloor\leq k-1$ for $i\in[m]$.
}
\label{fig:ex}
\end{figure}
Next, we choose a list of $(n-p)$ subsets $T_{p+1},\dots,T_{n}\subseteq [L+1]$, each of size $\ell$, such that $T_{p+1},\dots,T_n$ are ``evenly distributed.'' More specifically, for each $i\in[L+1]$, let $a_i$ denote the number of indices $j$ in $\{p+1,p+2,\dots,n\}$ such that $i\in T_j$. We choose sets $T_{p+1},\dots,T_n\subseteq [L+1]$ of size $\ell$ such that 
$\lfloor\frac{\ell(n-p)}{L+1}\rfloor\leq a_i\leq \lceil\frac{\ell(n-p)}{L+1}\rceil$ for $i\in [L+1]$.\footnote{We can construct $T_{p+1},\dots,T_n$ one by one. For each $p+1\leq i\leq n$, let $T_i=\{b_1,\dots,b_{\ell}\}\in \binom{[L+1]}{\ell}$ such that $a^{(i)}_{b_1},\dots,a^{(i)}_{b_{\ell}}$ are the smallest $\ell$ elements among $a^{(i)}_{1},\dots,a^{(i)}_{L+1}$. Here, $a^{(i)}_u$ for $u\in[L+1]$, denotes the number of indices $j\in\{p+1,p+2,\dots,i-1\}$ for which $u\in T_j$. Break ties arbitrarily. It is straightforward to see that $T_{p+1},\dots,T_n$ constructed in this way are ``evenly distributed.''}

Finally, we construct the product set $S$.
For each $i\in\{p+1,p+2,\dots n\}$, define $S_i:=\left\{\mathcal{C}(g_v)[i]:v\in T_i\right\}$, whose size is at most $|T_i|=\ell$.
For each $i\in[p]$, let $t_i$ be the unique integer in $[m]$ such that $i\equiv t_i \bmod{m}$, and then define $S_i:=\{\beta_1 \mathcal{C}(f_{t_i})[i],\dots,\beta_{\ell} \mathcal{C}(f_{t_i})[i]\}$. Here again, $|S_i|\leq \ell$.
Let $S:=S_1\times S_2\times\cdots\times S_n$. 

Consider arbitrary $v\in [L+1]$.
As $g_v\in F$, we may write $g_v=\sum_{j=1}^m q_j f_j$ with $q_1,\dots,q_m\in\{\beta_1,\dots,\beta_\ell\}$.
Then for each $i\in [p]$, by definition,
\[
\mathcal{C}(g_v)[i]=\left(g_v(\alpha_i),\dots,g_v(\gamma^{s-1}\alpha_i)\right)=\left(\sum_{j=1}^mq_jf_j(\gamma^{u}\alpha_i)\right)_{u\in\{0,\dots,s-1\}}.
\]
For any $u\in\{0,\dots,s-1\}$ and $t'\in [m]$ different from $t_i$, by construction, we know $Q_i(\gamma^u\alpha_i)=0$ and $Q_i(X)\mid f_{t'}(X)$, which implies $f_{t'}(\gamma^u\alpha_i)=0$. Therefore, for each $i\in [p]$,
\[
\mathcal{C}(g_v)[i]=(q_{t_i}f_{t_i}(\alpha_i),\dots,q_{t_i}f_{t_i}(\gamma^{s-1}\alpha_i))=q_{t_i}\mathcal{C}(f_{t_i})[i]\in S_i.
\]
For each $v\in[L+1]$, let $\mathcal{C}'(g_v)$ denote the segment of $\mathcal{C}(g_v)$ whose first $p$ indices are removed, we know $\mathrm{Agr}(\mathcal{C}'(g_v),S_{p+1}\times\cdots\times S_n)\ge a_v$. Therefore, 
\[
\mathrm{Agr}(\mathcal{C}(g_v),S)=p+\mathrm{Agr}(\mathcal{C}'(g_v),S_{p+1}\times\cdots\times S_n)
\geq p+a_v
\geq p+\left\lfloor\frac{\ell(n-p)}{L+1}\right\rfloor.
\]
It follows that 
\[
\mathrm{dist}(\mathcal{C}(g_v),S)=n-\mathrm{Agr}(\mathcal{C}(g_v),S)
\leq n-p-\left\lfloor\frac{\ell(n-p)}{L+1}\right\rfloor
\leq\left\lceil\frac{(L+1-\ell)(n-p)}{L+1}\right\rceil.
\]
The above holds for all $v\in [L+1]$, showing that the code cannot be $\left(\frac{\left\lceil\frac{(L+1-\ell)(n-p)}{L+1}\right\rceil}{n},\ell, L\right)$ list-recoverable. 
Therefore, if the code is $(\rho,\ell,L)$ list-recoverable, then
\[
\rho\leq \frac{\left\lceil\frac{(L+1-\ell)(n-p)}{L+1}\right\rceil-1}{n}\leq \frac{L+1-\ell}{L+1}\left(1-\frac{p}{n}\right).
\]
This completes the proof of \cref{thm:listrec}. 


\subsection{Tightness of the Improved Upper Bound in the Case of $(\ell, L)=(2,3)$}

The bound above is derived from a highly structured, hypercube-like arrangement of polynomials, which we believe represents the ``worst-case'' configuration within the linear space of messages. Accordingly, we conjecture that \cref{thm:bound_clean} is (nearly) tight when $L+1$ takes the form $\ell^a$, allowing the $L+1$ messages to fully constitute a hypercube-like list of candidate messages. 
\begin{conjecture}\label{conj:tight}
For any $\ell\ge 2$, generator $\gamma$ of $\mathbb{F}^{\times}_q$ and $L+1=\ell^a$ where $a\in\mathbb{N}^{\ge 2}$, \cref{thm:bound_clean} is almost tight.
Formally, for any constants $\epsilon>0, \ell\ge 2, L+1=\ell^a, a\in\mathbb{N}^{\ge 2}, R\leq \frac{a-1}{a}$ and generator $\gamma$ of $\mathbb{F}^{\times}_q$ there exists a constant $C$ such that: if $s\ge C$, $\frac{k-1}{s}\ge a$, and $n$ is suffciently large, then any rate 
$R=\frac{k}{sn}$ folded Reed--Solomon code $\mathsf{FRS}^{(s,\gamma)}_{n,k}(\alpha_1,\alpha_2,\dots,\alpha_n)$ with appropriate evaluation points in $\mathbb{F}_q$ is $(\rho-\epsilon,\ell, L)$ list-recoverable, where
\[\rho=\frac{L+1-\ell}{L+1}\left(1-\frac{aR}{a-1}\right).\]
\end{conjecture}
Then, as the first step to resolve \cref{conj:tight}, we will show that when $\ell=2,L=3$, \cref{conj:tight} is true. In this case, the parameter $a=2$. It suffices to prove  \cref{thm:special}. The proof is based on extending the framework we built in \cref{sec:frs} for list-decoding.
\begin{theorem}[Restatement of \cref{thm:special}]
For any $s,n,k\in\mathbb{N}^+,q>n,s\ge 3$ and generator $\gamma$ of $\mathbb{F}_q^\times$. The folded RS code $\mathsf{FRS}_{n,k}^{(s,\gamma)}(\alpha_1,\alpha_2,\dots,\alpha_n)$ with appropriate evaluation points in $\mathbb{F}_q$ is $\left(\frac{1}{2}-\frac{sR}{s-2}, 2, 3\right)$ list-recoverable.
\end{theorem}
\begin{proof}
First, we need to generalize the geometric agreement hypergraph \cref{def:agr_graph} used for list-decoding to a list-recovery counterpart. To achieve this we define geometric table-agreement hypergraph.
\begin{definition}[Geometric table-agreement hypergraph based on FRS codes]
Given a $(s,\gamma)$-folded Reed--Solomon code $\mathsf{FRS}^{(s,\gamma)}_{n,k}(\alpha_1,\dots,\alpha_n)\subseteq \left(\mathbb{F}^s_q\right)^n$ where $(\alpha_1,\dots,\alpha_n)$ is an appropriate sequence, a received product set $S=S_1\times\cdots\times S_n\in(2^{\mathbb{F}^s_q})^n$, and $\ell$ vectors $\vec{f_1},\vec{f_2},\dots,\vec{f_\ell}\in \mathbb{F}_q^k$, we can define the geometric table-agreement hypergraph $(\mathcal{V},\mathcal{E})$ with vertex set $\mathcal{V}:=\left\{\vec{f_1},\vec{f_2},\dots,\vec{f_\ell}\right\}$ and a tuple of $\sum_{i=1}^n|S_i|$ hyperedges $\mathcal{E}:=\bigcupdot_{i=1}^n\mathcal{E}_i$, where
$\mathcal{E}_i:=\{e_{i,1},\dots,e_{i,|S_i|}\subseteq\mathcal{V}\}$ and $e_{i,j}:=\left\{\vec{f_t}\in\mathcal{V} : s_{i,j}=\mathcal{C}(f_t)[i]\right\}$. Here $s_{i,j}\in\mathbb{F}^s_q$ means the $j$-th element in $S_i$.
\end{definition}
We can also easily define the weight function for geometric table-agreement hypergraphs and their hyperedges.
\begin{definition}[Weight]
Given a geometric table-agreement hypergraph $(\mathcal{V},\mathcal{E})$ as above, we define the weight $\mathrm{wt}(\mathcal{V},\mathcal{E}):=\sum^n_{i=1}\sum_{j=1}^{|S_i|}\mathrm{wt}(e_{i,j})$, where $\mathrm{wt}(e_{i,j}):=\max\bigg(|e_{i,j}|-1,0\bigg)$.
\end{definition}
Suppose by contradiction that there exists four distinct polynomials $f_1,f_2,f_3,f_4\in\mathbb{F}_q[x]_{<k}$ and a received product set $S=S_1\times\cdots\times S_n\in\binom{\mathbb{F}^s_q}{\leq 2}^n$ such that for each $i\in[4]$, the codeword $c_i:=\mathcal{C}(f_i)$ of $f_i$ has relative hamming distance at most $\left(\frac{1}{2}-\frac{sR}{s-2}\right)$ from $S$. Then for any $i\in[4]$, there is a subset $I_i\subseteq [n]$ where $|I_i|\ge \frac{n}{2}+\frac{k}{s-2}$, such that for each $t\in I_i$, $c_{i}[t]\in S_t$. For each $j\in[n]$, without loss of generality suppose $S_j=\{s_{j,1},s_{j,2}\in\mathbb{F}^s_q\}$, then for each $k\in\{1,2\}$ we define $e_{j,k}=\left\{\vec{f}_i\colon j\in I_i,c_i[j]=s_{j,k}\right\}$ as the set of polynomials whose codewords match with $s_{j,k}$ on position $j$ in this bad list. Let $(\mathcal{V},\mathcal{E})$ denote the corresponding geometric table-agreement hypergraph, where $\mathcal{V}=\left\{\vec{f}_1,\vec{f}_2,\vec{f}_3,\vec{f}_{4}\right\}$, and $\mathcal{E}=(e_{1,1},e_{1,2},\dots,e_{n,1},e_{n,2})$.
The weight of this hypergraph can be bounded by
\begin{equation*}
\mathrm{wt}(\mathcal{V},\mathcal{E})\ge \sum_{j=1}^n\sum_{k=1}^2\bigg(|e_{j,k}|-1\bigg)\ge\left(\sum^{4}_{i=1}|I_i|\right)-2n\ge\frac{4k}{s-2}.
\end{equation*}
For each hyperedge $e_{j,k}\subseteq \mathcal{V},j\in[n],k\in\{1,2\}$, we choose an arbitrary spanning tree $T_{j,k}$ in the complete graph supported on the vertex-set $e_{j,k}$. Let $G=(\mathcal{V},E(G))$ be an undirected multi-graph (each edge connects two vertices) where $E(G)=\bigcupdot_{j\in[n],k\in\{1,2\}}E(T_{j,k})$. We can observe that
\begin{itemize}
\item[(1)] $|E(G)|=\mathrm{wt}(\mathcal{V},\mathcal{E})$
\item[(2)] For any subset of vertices $\mathcal{V}'\subseteq \mathcal{V}$, there is $\mathrm{wt}(\mathcal{V}',\mathcal{E}|_{\mathcal{V}'})\ge |E(G[\mathcal{V}'])|$, where $G[\mathcal{V'}]$ denotes the induced subgraph of $G$ on $\mathcal{V}'$.
\end{itemize}
There are $4$ different subsets $\mathcal{V}_1,\mathcal{V}_2,\mathcal{V}_3,\mathcal{V}_4\subseteq \mathcal{V}$ of vertices, each with size $3$. We can lower bound the sum of their `weights'
\[\sum_{i=1}^4\mathrm{wt}(\mathcal{V}_i,\mathcal{E}|_{\mathcal{V}_i})\ge \sum_{i=1}^4|E(G[\mathcal{V}_i])|=2|E(G)|=2\mathrm{wt}(\mathcal{V},\mathcal{E})\ge \frac{8k}{s-2}\]
The second derivation $\sum_{i=1}^4|E(G[\mathcal{V}_i])|=2|E(G)|$ above comes from the fact that every edge $e=(u_e,v_e)\in E(G)$ is counted exactly twice in $\sum_{i=1}^4|E(G[\mathcal{V}_i])|$ since there are exactly two vertex subsets of $G$ with size $3$ contain both $u_e$ and $v_e$.

Therefore, there must be some $\mathcal{V}_r, |\mathcal{V}_r|=3, r\in[4]$ such that $\mathrm{wt}(\mathcal{V}_r,\mathcal{E}|_{\mathcal{V}_r})\ge \frac{2k}{s-2}$. Moreover, since there are only three vertices in $\mathcal{V}_r$, for any index $i\in[n]$, at most one of $e_{i,1},e_{i,2}$ contributes positive weights to $\mathrm{wt}(\mathcal{V}_r,\mathcal{E}|_{\mathcal{V}_r})$. That's because they are two disjoint hyperedges by definition. Therefore, for each index $i\in[n]$ we can select at most one hyperedge $e'_i\in\{e_{i,1}|_{\mathcal{V}_r},e_{i,2}|_{\mathcal{V}_r}\}$ such that $\mathrm{wt}(\mathcal{V}_r,\mathcal{E}':=(e'_1,\dots,e'_n))=\mathrm{wt}(\mathcal{V}_r,\mathcal{E}|_{\mathcal{V}_r})\ge\frac{2k}{s-2}$. Since $(\mathcal{V}_r, \mathcal{E}')$ is just a normal geometric agreement hypergraph defined in \cref{def:agr_graph} with a single hyperedge for each index (not a table-agreement one), we can resort to \cref{lem:non_distinct} on $(\mathcal{V}_r,\mathcal{E}')$ and conclude that $f_1, f_2, f_3, f_4$ cannot be distinct, which leads to a contradiction.
\end{proof}

\section{Future Directions}
We highlight several directions for future research.
\begin{enumerate}
    \item  \emph{Deterministic list-decoding algorithms.}
Most recently, Goyal, Harsha, Kumar, and Shankar \cite{goyal2023fast} provided near-linear time (randomized) list-decoding algorithms for both folded RS and univariate multiplicity codes. While we 
have proved that all ``appropriate'' folded RS codes are list-decodable up to capacity, it remains an interesting question to obtain an efficient and deterministic list-decoding algorithm for folded RS codes that achieve list-decoding capacity.

\item \emph{List-recoverability \& better constructions of lossless condensers.}
A natural follow-up open problem is to establish better list-recoverability of folded RS and univariate multiplicity codes, or their relatives, improving upon the best known bound from \cite{kopparty2018improved,kopparty2023improved,tamo2024tighter}. In \cref{sec:list_recover}, we have proposed a new bound in \cref{thm:bound_clean} and conjectured in \cref{conj:tight} that it is tight. We are curious whether this conjecture is true. If it is false, what should be the correct tight bound?

One of the most important applications of list-recoverability is the construction of seeded condensers. The seminal paper \cite{guruswami2009Unbalanced} established a connection between list-recoverable codes and both lossless and lossy seeded condensers. The line of work \cite{guruswami2009Unbalanced,ta2012better,kts22} exploited this connection and used explicit Parvaresh–Vardy codes, folded RS codes, and univariate multiplicity codes with good list-recoverability to achieve the best known explicit constructions of lossless (and lossy) seeded condensers. These results typically focus on a large list size $L$. We hope that the ideas and techniques in our paper will be helpful for the construction of better seeded condensers. 
\item \emph{Explicit RS codes achieving list-decoding capacity.}  Although there is a long line of works \cite{rudra2014every, shangguan2020combinatorial, guo2022improved, ferber2022list, goldberg2022list,brakensiek2023generic,GZ23, alrabiah2023randomly} showing that randomly punctured RS codes over linear-sized alphabets are list-decodable up to list-decoding capacity with high probability, little is known about constructing explicit RS codes that achieve list-decoding capacity. In fact, all known ``explicit'' RS codes beyond the Johnson radius \cite{shangguan2020combinatorial,roth22,BDG24} either require doubly exponential-sized alphabets or work only 
for very restricted parameter regimes. An important open problem is the explicit construction of RS codes over polynomial-sized alphabets that are list-decodable beyond the Johnson radius or even achieve list-decoding capacity. In fact, even achieving an exponential alphabet size would be a breakthrough, since the time complexity would be polylogarithmic in the alphabet size, which is polynomial. Since our paper focuses on explicit folded RS codes that achieve list-decoding capacity, we hope some techniques and ideas presented here may be useful in constructing the desired RS codes.
\end{enumerate}
\section*{Acknowledgments}
The authors would like to thank Joshua Brakensiek, Mahdi Cheraghchi, Manik Dhar, Sivakanth Gopi, Zeyu Guo, Venkatesan Guruswami, Jiyou Li, and Chong Shangguan for many helpful discussions and suggestions. In particular, the authors would like to extend their special thanks to Venkatesan Guruswami for highlighting the connection between \cref{thm:lower_bound_root} and {\cite[Theorem 14]{guruswami2016explicit}}, which leads to more general results stated in \cref{append:generalize}. 

This work was initiated while the two authors were visiting the Simons Institute for the Theory of Computing, supported by DOE grant DE-SC0024124.  The authors would like to thank the institute for its support and
hospitality. Yeyuan Chen is partially supported by the National Science Foundation under Grants No.\ CCF-2107345 and CCF-2236931.
Zihan Zhang is supported by the National Science Foundation under Grants No.\ CCF-2440926.

{\bibliography{main}} 

\listoffixmes
\appendix
\section{Univariate Multiplicity Codes Achieve Relaxed Generalized Singleton Bounds}\label{sec:mul}
In this appendix, following a similar approach to that in \cref{sec:frs}, we prove our main results on the list-decodability of univariate multiplicity codes (\cref{thm:main_mul} and \cref{cor:main_mul}).

\subsection{Properties of Hasse derivatives}
Recall that the $i$-th Hasse derivative $f^{(i)}(X)$ of a polynomial $f(X)$ is defined as the coefficient of $Z^i$ in the expansion $f(X+Z)=\sum_{i\in\mathbb{N}}f^{(i)}(X)Z^i$.
In this section, we provide some known facts about Hasse derivatives, which will be used in the proofs.
\begin{proposition}[Basic properties of Hasse derivatives, see \cite{hirs_book,kopparty13}]\label{pro:hasse}
Let $\mathbb{F}$ be a field and $i,j\in\mathbb{N}$. Let $f(X),g(X)\in\mathbb{F}[X]$ be two polynomials. Then we have the following:
\begin{itemize}
\item $f^{(i)}(X)+g^{(i)}(X)=(f+g)^{(i)}(X)$
\item $af^{(i)}(X)=(af)^{(i)}(X)$\text{ } for any $a\in\mathbb{F}$
\item $\left(f^{(i)}\right)^{(j)}(X)=\binom{i+j}{i}f^{(i+j)}(X)$
\item $\left(f\cdot g\right)^{(i)}(X)=\sum^i_{k=0}f^{(k)}(X)g^{(i-k)}(X)$
\end{itemize}
\end{proposition}
We also need the following fact to adapt the proof of \Cref{thm:lower_bound_root} to univariate multiplicity codes.
\begin{claim}\label{clm:two_derivative}
Let $\mathbb{F}$ be a field. Let $f(X)\in\mathbb{F}[X]$ and $\alpha\in\mathbb{F}$. If $f^{(i)}(\alpha)=0$ for $i\in\{0,1,\dots,s-1\}$,  then for any $j\in\{0,1,\dots,s-1\}$, $\alpha$ is a root of $f^{(j)}(X)$ with multiplicity at least $s-j$.
\end{claim}
\begin{proof}
Consider any $j\in\{0,\dots,s-1\}$. It suffices to prove that $(X-\alpha)^k$ divides $f^{(j)}(X)$ for all $k\in\{0,\dots, s-j\}$. We prove this claim by induction on $k$. 
The claim trivially holds for $k=0$.

Assume that the claim holds for some $k\in\{0,\dots,s-j-1\}$, and we now prove that it holds for $k+1$ as well.
By the induction hypothesis, $f^{(j)}(X)=(X-\alpha)^{k}g(X)$ for some polynomial $g(X)\in\mathbb{F}[X]$. By \Cref{pro:hasse}, we have
\begin{equation*}
\binom{j+k}{j}f^{(j+k)}(X)=\left(f^{(j)}\right)^{(k)}(X)=\left((X-\alpha)^kg(X)\right)^{(k)}=\sum^{k}_{t=0}\binom{k}{t}(X-\alpha)^tg^{(t)}(X)
\end{equation*}
Evaluating $X=\alpha$, the above equation implies
\begin{equation*}
\binom{j+k}{j}f^{(j+k)}(\alpha)=g(\alpha)
\end{equation*}
Since $f^{(j+k)}(\alpha)=0$, we have $g(\alpha)=0$. We conclude that $(X-\alpha)$ divides $g(X)$, and therefore $(X-\alpha)^{k+1}$ divides $f^{(j)}(X)$.
\end{proof}
\subsection{Geometric agreement hypergrahs and geometric polynomials based on univariate multiplicity codes}
In this section, to adapt to univariate multiplicity code case, we slightly modify necessary definitions of geometric agreement hypergraph and geometric polynomials, and state necessary lemmas.\begin{definition}[Geometric agreement hypergraph based on univariate multiplicity codes]\label{def:mult_agr_graph}
 Given an order-$s$ univariate multiplicity code $\mathsf{MULT}^{(s)}_{n,k}(\alpha_1,\dots,\alpha_n)\subseteq \left(\mathbb{F}^s_p\right)^n$ where $\alpha_1,\dots,\alpha_n$ are distinct, a received word $\vec{y}\in(\mathbb{F}^s_p)^n$, and $\ell$ vectors $\vec{f_1},\vec{f_2},\dots,\vec{f_\ell}\in \mathbb{F}_p^k$, we can define the \emph{geometric agreement hypergraph} $(\mathcal{V},\mathcal{E})$ with vertex set $\mathcal{V}:=\left\{\vec{f_1},\vec{f_2},\dots,\vec{f_\ell}\right\}$ and a tuple of $n$ hyperedges $\mathcal{E}:=\{e_1,e_2,\dots,e_n\}$, where
$e_i:=\left\{\vec{f_j}\in\mathcal{V} : \vec{y}[i]=\mathcal{M}(f_j)[i]\right\}.$
\end{definition}
\begin{definition}[Classical Wronskian, See \cite{guruswami2016explicit}] Let $f_1(X), \ldots, f_s(X) \in \mathbb{F}_p[X]$. We define their  Wronskian $W\left(f_1, \ldots, f_s\right)(X) \in \left(\mathbb{F}_p[X]\right)^{s\times s}$ by
$$
W\left(f_1, \ldots, f_s\right)(X) \stackrel{\text { def }}{=}\left(\begin{array}{ccc}
f_1(X) & \cdots & f_s(X) \\
f^{(1)}_1(X) & \cdots & f^{(1)}_s(X) \\
\vdots & \ddots & \vdots \\
f^{(s-1)}_1\left(X\right) & \cdots & f^{(s-1)}_s\left(X\right)
\end{array}\right).
$$  
\end{definition}
\begin{lemma}[Classical Wronskian criterion for linear independence, See \cite{guruswami2016explicit,muir2003book}]\label{lem:classical_wrons_ind} Given $k<\mathrm{char}(\mathbb{F}_q)$, and let $\vec{f_1}, \ldots, \vec{f_s} \in \mathbb{F}^k_q$. Then $\vec{f_1}, \ldots, \vec{f_s}$ are linearly independent over $\mathbb{F}_q$ if and only if the classical Wronskian determinant $\operatorname{det} W\left(f_1, \ldots, f_s\right)(X) \neq 0$.
\end{lemma}
\begin{definition}[Geometric polynomial based on classical Wronskians]
    Given $L$ non-zero vectors $\vec{f}_1,\vec{f}_2,\dots,\vec{f}_L\in\mathbb{F}_p^k$ such that 
$\dim_{\mathbb{F}_p}\left(\mathrm{Span}_{\mathbb{F}_p}\left\{\vec{f}_1,\vec{f}_2,\dots,\vec{f}_L\right\}\right)=\ell\in[L]$. Then we define our geometric polynomial
$\widetilde{V}_{\left\{\vec{f}_i\right\}_{i\in L}}(X)$ as the following monic polynomial
\[\lambda_{i_1,i_2,\dots,i_\ell}\cdot\det W(f_{i_1},f_{i_2},\dots,f_{i_\ell})(X),\]
where $\lambda_{i_1,i_2,\dots,i_\ell}\in\mathbb{F}_p^\times$, $i_{j}\in[L]$ for all $j\in[\ell]$, and $\left\{\vec{f}_{i_1},\vec{f}_{i_2},\dots,\vec{f}_{i_\ell}\right\}$ forms a $\mathbb{F}_p$-basis of the space $\mathrm{Span}_{\mathbb{F}_p}\left\{\vec{f}_1,\vec{f}_2,\dots,\vec{f}_L\right\}$.
\end{definition}
\begin{lemma}\label{lem:classical_geometric_poly}
  Geometric polynomials based on classical Wronskians are well-defined.  
\end{lemma}
\begin{proof}
Let $V=\mathrm{Span}_{\mathbb{F}_p}\left\{\vec{f}_1,\vec{f}_2,\dots,\vec{f}_L\right\}$, for any two basis $\{\vec{u}_i\}_{i\in[\ell]},\{\vec{u}'_i\}_{i\in[\ell]}$ of $V$, from \Cref{pro:hasse}, there exists a non-singular matrix $A\in\mathbb{F}^{\ell\times\ell}_p$ satisfying $W(u_1,\dots,u_{\ell})=W(u'_1,\dots,u'_{\ell})A$. Using \Cref{lem:classical_wrons_ind}, the proof is similar to \Cref{lem:geometric_poly}.
\end{proof}
\begin{theorem}[{Alternatively stated in \cite[Theorem 17]{guruswami2016explicit}}]\label{thm:mult_lower_bound_root}
Given $L$ distinct non-zero polynomials $f_1,\dots,f_L\in\mathbb{F}^k_p$ with degree at most $k-1$. Let $(\mathcal{V},\mathcal{E})$ be a geometric agreement hypergraph over $\mathcal{V}=\left\{0,\vec{f}_1,\dots,\vec{f}_L\right\}$ where $\mathcal{E}=\big\{e_1,\dots,e_n\subseteq \mathcal{V}\big\}$, it follows that $P(X)=\widetilde{V}_{\{f_i\}_{i\in L}}(X)$ has at least $(s-\ell+1)\sum_{i=1}^n\widetilde{\dim}_{\mathbb{F}_p}\left(e_i\right)$ roots counting multiplicity where $\ell=\dim\left(\mathrm{Span}_{\mathbb{F}_p}\left\{\vec{f}_1,\vec{f}_2,\dots,\vec{f}_L\right\}\right)$.
\end{theorem}
\begin{proof}
Similar to \Cref{thm:lower_bound_root}, given any $e_i,i\in[n]$, let's focus on a fixed $e_i$ with $\widetilde{\dim}_{\mathbb{F}_p}\left(e_i\right)=t$. It suffices to prove that $\alpha_i$ is a root of $\widetilde{V}_{\{f_i\}_{i\in L}}(X)$ with multiplicity $(s-\ell+1)t$.

There must exist $\vec{g}_0,\dots,\vec{g}_{t}\in e_i$ such that $\left\{\vec{h}_1=\vec{g}_1-\vec{g}_0,\dots,\vec{h}_t=\vec{g}_t-\vec{g}_0\right\}$ are linear independent. Let $V=\Span_{\mathbb{F}_p}\left\{\vec{f}_1,\dots,\vec{f}_L\right\}$, we can arbitrarily extend $\left\{\vec{h}_u\right\}_{u\in [t]}$ to a basis $\left\{\vec{h}_u\right\}_{u\in[\ell]}$ of $V$.

By \Cref{lem:classical_geometric_poly}, $Q(X)=\det W(h_1,\dots,h_{\ell})(X)=\lambda \widetilde{V}_{\left\{\vec{f}_1,\dots,\vec{f}_L\right\}}(X)$ for some $\lambda\in\mathbb{F}^\times_p$. Then, from \Cref{def:mult_agr_graph} and \Cref{pro:hasse}, we know for any $v\in[t],r\in\{0\}\cup[s-1]$, it follows that $h^{(r)}_v(\alpha_i)=g^{(r)}_v(\alpha_i)-g^{(r)}_0(\alpha_i)=0$. By \Cref{clm:two_derivative}, it implies $(X-\alpha)^{s-\ell+1}$ divides $h^{(u)}_v(X)$ for all $u\in\{0\}\cup[\ell-1]$. Since the $(u,v)$-entry of $ W(h_1,\dots,h_{\ell})(X)$ is $h^{(u-1)}_v(X)$, we know that $(X-\alpha_i)^{s-\ell+1}$ is a factor for each entry in the first $t$ columns of $W(h_1,\dots,h_{\ell})(X)$. Consider computing $\det W(h_1,\dots,h_{\ell})(X)$ by expanding the matrix along the first $t$ columns, we can see that $\alpha_i$ is a root of $\det W(h_1,\dots,h_{\ell})(X)$ with multiplicity at least $(s-\ell+1)t$, which also implies its same multiplicity in $P(X)$. 
\end{proof}
\subsection{Putting it together: based on \cref{sec:frs}}
With the above adaptions to univariate multiplicity codes, we can use basically the same proof as in \Cref{thm:main} to get the list-decodability of univariate multiplicity codes.
\begin{theorem}[Restatement of \Cref{thm:main_mul}]
Let $p$ be a prime number. For any list size $L\ge1$, if $s,n,k\in\mathbb{N}^+$ and $\alpha_1,\alpha_2,\dots,\alpha_n\in\F_p$ are distinct, then $\mathsf{MULT}_{n,k}^{(s)}(\alpha_1,\alpha_2,\dots,\alpha_n)$ over the alphabet $\mathbb{F}_p^s$ is $\left(\frac{L}{L+1}\left(1-\frac{sR}{s-L+1}\right),L\right)$ list-decodable.
\end{theorem}
\begin{corollary}[Restatement of \Cref{cor:main_mul}]\label{cor:exact_mult}
    For any $\epsilon>0$, $N>k\ge1$, prime $p>k$, $L>\frac{1-R-\epsilon}{\epsilon}$,  $s>\frac{L(L-1)R}{\epsilon L-(1-R-\epsilon)}+L-1$ where $s|N$ and $R=\frac{k}{N}$. Let 
$n=\frac{N}{s}$ and  $\alpha_1,\alpha_2,\dots,\alpha_n\in\F_p$ be distinct. Then the order-$s$ univariate multiplicity code $\mathsf{MULT}_{n,k}^{(s)}(\alpha_1,\alpha_2,\dots,\alpha_n)$ over the alphabet $\mathbb{F}_p^s$ is $\big(1-R-\epsilon,L\big)$ list-decodable.
\end{corollary}
\section{Subspace Designable Codes Exhibit Near-Optimal List-Decodability}\label{append:generalize}
In this appendix, following procedures almost the same as those outlined in \cref{sec:frs} and \cref{sec:mul}, we aim to 
establish a more general result that shows explicit ``strong subspace designable codes'' achieve list deocoding capacity. The core idea is to use a more generalized framework called ``subspace design'' built in \cite{guruswami2016explicit} to replace the ``geometric polynomials'' in the previous proofs, and uses our new techniques developed in \cref{sec:loss} to reduce the list size. Using the ``subspace design'' language, this section also gives a more modular presentation of our proofs in \cref{sec:frs,sec:mul}. First, we need to introduce the notion ``strong subspace design.'' 
\begin{definition}[{Strong Subspace Design, \cite[Definition 3]{guruswami2016explicit}}] A collection $\mathcal{H}$ of $\mathbb{F}$-linear subspaces $H_1,\dots,H_n\subseteq \mathbb{F}^k$ is called an $(\ell, A)$ strong subspace design over $\mathbb{F}$, if for every $\mathbb{F}$-linear space $W\subseteq \mathbb{F}^k$ of dimension $\ell$, we have
\[\sum_{i=1}^n\dim_{\mathbb{F}}\left(H_i\cap W\right)\leq A.\].
\end{definition}
For any $\mathbb{F}$-linear code over $\mathbb{F}^s,s\ge 1$ with block length $n$, we can actually define a corresponding collection of $n$ subspaces within its message space. In our generalization, the parameters these subspaces achieve as `strong subspace design' will play a role as the degree upper bound of non-zero geometric polynomials in the previous proofs. From this relation between codes and strong subspace design, we define the notion `strong subspace designable codes' below.
\begin{definition}[Strong Subspace Designable Codes]\label{def:designcodes} For any $s\ge 1$, given an $\mathbb{F}$-linear code $C\subseteq \left(\mathbb{F}^s\right)^n$ with message length $k$ and block length $n$, we use $\mathcal{C}\colon \mathbb{F}^k\to\left(\mathbb{F}^s\right)^n$ to denote the $\mathbb{F}$-linear encoder of $C$. For any $i\in[n]$, let $H_i\subseteq \mathbb{F}^k$ denote the $\mathbb{F}$-linear subspace such that for any message $f\in\mathbb{F}^k$, there is $\mathcal{C}(f)_i=0$ iff $f\in H_i$. We say $C$ is a $(\ell,A)$-strong subspace designable code if $\mathcal{H}:=\{H_1,\dots,H_n\}$ is an $(\ell, A)$-strong subspace design.
\end{definition}
Then, the following theorem serves as the counterpart of \cref{thm:lower_bound_root} (\cref{thm:mult_lower_bound_root}) in our generalization. It uses the affine dimension to build a lower bound of the interested quantity. We need to recall the definition of geometric agreement  hypergraph in \cref{def:agr_graph}. Here we need this notion based on our interested code $C$, which means the hyperedge $e_i,i\in[n]$ here implies the $C$-codewords of  polynomials in $e_i$ match on the $i$-th index.
\begin{theorem}\label{thm:lowerbounddesign}
Let $C\subseteq (\mathbb{F}^s)^n$ and $H_1,\dots,H_n\subseteq \mathbb{F}^k$ be the code and linear subspace defined as in \Cref{def:designcodes}. Given $L$ distinct non-zero vectors $\vec{f}_1,\vec{f}_2,\dots,\vec{f}_L\in\mathbb{F}^k$. Let $(\mathcal{V},\mathcal{E})$ be a geometric agreement hypergraph based on $C$ over $\mathcal{V}=\left\{0,\vec{f}_1,\dots,\vec{f}_L\right\}$ where $\mathcal{E}=\big\{e_1,\dots,e_n\subseteq \mathcal{V}\big\}$, it follows that $\sum_{i=1}^n\dim_{\mathbb{F}}\left(H_i\cap W\right)\ge\sum_{i=1}^n\widetilde{\dim}_{\mathbb{F}}(e_i)$ where $W=\mathrm{Span}_{\mathbb{F}}\left\{\vec{f}_1,\vec{f}_2,\dots,\vec{f}_L\right\}$.
\end{theorem}
\begin{proof}
Given any $i\in[n]$, our goal is to prove $\dim_{\mathbb{F}}\left(H_i\cap W\right)\ge \widetilde{\dim}_{\mathbb{F}}(e_i)$. Let $\ell=\widetilde{\dim}_{\mathbb{F}}(e_i)>0$, by \cref{def:agr_graph} and \cref{fact:transform_tildedim}, there exists $\mathbb{F}$-linear independent vectors  $\vec{g}_1:=\vec{f}_{i_1}-\vec{f}_{i_0},\dots,\vec{g}_{\ell}:=\vec{f}_{i_{\ell}}-\vec{f}_{i_0}$ such that for any $j\in[\ell]$, $\mathcal{C}(\vec{g}_j)_i=0$. Therefore, $\vec{g_1},\dots,\vec{g_{\ell}}$ are $\mathbb{F}$-linear independent within $H_i\cap W$, which implies $\dim_{\mathbb{F}}\left(H_i\cap W\right)\ge \ell=\widetilde{\dim}_{\mathbb{F}}(e_i)$.   
\end{proof}
Now it's time to put them together. We will use the similar proof strategy as in \cref{sec:puttogether} to derive the generalized result \cref{thm:generalize}. \cref{thm:bound_loss} used in the following proof is the central technical lemma that are different from all previous work. It analyzes geometric structures of polynomials interested and derives a lower bound on the sum of affine dimensions of hyperedges. The details are illustrated in \cref{sec:loss}. 
\begin{lemma}\label{lem:non-distinctdesign}
Given a $\F$-linear code $C\subseteq (\mathbb{F}^s)^n$ of length $n$ and rate $R={k}/{sn}$. Assume that $C\subseteq (\mathbb{F}^s)^n$ is a $\left(\ell, \frac{\ell{(k-1)}}{s-\ell+1}\right)$-strong subspace designable code for all $\ell\leq {s}$. Then,
for any $2\leq m\leq s+1$ vectors $\vec{f}_1,\dots,\vec{f}_m\in\mathbb{F}^k$ and a 
$C$-based geometric agreement hypergraph $(\mathcal{V},\mathcal{E})$ over $\mathcal{V}=\left\{\vec{f}_1,\dots,\vec{f}_m\right\}$. If $\mathrm{wt}(\mathcal{V},\mathcal{E})\ge \frac{(m-1)k}{s-m+2}$, then $\vec{f}_1,\dots,\vec{f}_m$ cannot be distinct.
\end{lemma}
\begin{proof} This proof follows from the similar framework as in \cref{lem:non_distinct}. Since $\mathrm{wt}(\mathcal{V},\mathcal{E})\ge \frac{(|\mathcal{V}|-1)k}{s-|\mathcal{V}|+2}$, there must exist a minimal subset $\mathcal{V}_0\subseteq \mathcal{V}$ with $|\mathcal{V}_0
|\ge 2$ and 
$\mathcal{V}_0$ satisfies the following conditions. 
\begin{itemize}
\item $\mathrm{wt}(\mathcal{V}_0,\mathcal{E}|_{\mathcal{V}_0})\ge\frac{(|\mathcal{V}_0|-1)k}{s-|\mathcal{V}_0|+2}$
\item For any proper subset $\mathcal{H}\subsetneq \mathcal{V}_0$ with $|\mathcal{H}|\ge 2$, $\mathrm{wt}(\mathcal{H},\mathcal{E}|_{\mathcal{H}})<\frac{(|\mathcal{H}|-1)k}{s-|\mathcal{H}|+2}$.
\end{itemize}
Let $m':=|\mathcal{V}_0|\ge 2$ and $\mathcal{V}_0=\{\vec{f}_{i_1},\dots,\vec{f}_{i_{m'}}\}$. Suppose by contradiction that $\vec{f}_1,\dots,\vec{f}_m$ are distinct, then $\vec{g}_1:=\vec{f}_{i_1}-\vec{f}_{i_1},\dots,\vec{g}_{m'}:=\vec{f}_{i_{m'}}-\vec{f}_{i_1}$ are also distinct. Moreover, by the definition of $\mathcal{V}_0$, there must exist a corresponding $C$-based geometric agreement hypergraph $(\mathcal{V}',\mathcal{E}')$ such that
\begin{itemize}
\item $\mathcal{V}'=\{\vec{g}_1=0,\vec{g}_2,\dots,\vec{g}_{m'}\},\mathcal{E}'=\{e_1,\dots,e_n\subseteq \mathcal{V}' \}$ where $\vec{g}_1=0,\dots,\vec{g}_{m'}$ are distinct and $2\leq m'\leq m$.
\item $\mathrm{wt}(\mathcal{V}',\mathcal{E}')\ge\frac{(|\mathcal{V}'|-1)k}{s-|\mathcal{V}'|+2}$
\item For any proper subset $\mathcal{H}\subsetneq \mathcal{V}'$ with $|\mathcal{H}|\ge 2$, $\mathrm{wt}(\mathcal{H},\mathcal{E}'|_{\mathcal{H}})<\frac{(|\mathcal{H}|-1)k}{s-|\mathcal{H}|+2}$.
\end{itemize}
Let  $W:=\Span_{\mathbb{F}}\{\vec{g}_2,\dots,\vec{g}_{m'}\}, \ell:=\dim_{\mathbb{F}}W, 1\leq \ell\leq m'-1$ and $H_1,\dots,H_n$ be the subspaces defined in \cref{def:designcodes} over $C$. By \Cref{thm:lowerbounddesign}, we know 
\[\sum_{i=1}^n\dim_{\mathbb{F}}\left(H_i\cap W\right)\ge \sum_{i=1}^n\widetilde{\dim}_{\mathbb{F}_q}\left(e_i\right)\ge \bigg(\mathrm{wt}(\mathcal{V}',\mathcal{E}')-\sum^n_{i=1}\mathrm{Loss}(e_i)\bigg)\]
By the weight lower bound and \Cref{thm:bound_loss}, we have:
\begin{equation*}
\sum_{i=1}^n\dim_{\mathbb{F}}\left(H_i\cap W\right)\ge \left(\mathrm{wt}(\mathcal{V}',\mathcal{E}')-\sum^n_{i=1}\mathrm{Loss}(e_i)\right)\ge
\left(\frac{(m'-1)k}{s-m'+2}-\frac{(m'-1-\ell)k}{s-m'+2}\right)
\ge\frac{\ell k}{s-\ell+1}
\end{equation*}
This contradicts the assumption that $C$ is a $\left(\ell,\frac{\ell(k-1)}{s-\ell+1}\right)$-strong subspace designable code. We conclude that $\vec{f}_1,\dots,\vec{f}_m$ cannot be distinct.
\end{proof}
\begin{theorem}[A generalization of \cref{thm:main} and \cref{thm:main_mul}]\label{thm:generalize}
Given a $\F$-linear code $C\subseteq (\mathbb{F}^s)^n$ of block length $n$ and rate $R={k}/{sn}$. Assume that $C\subseteq (\mathbb{F}^s)^n$ is a $\left(\ell, \frac{\ell{(k-1)}}{s-\ell+1}\right)$-strong subspace designable code for all $\ell\leq {s}$. Then, $C$ is $\left(\frac{L}{L+1}\left(1-\frac{sR}{s-L+1}\right),L\right)$ (average-radius) list-decodable for any $L\leq {s}$.
\end{theorem}
\begin{proof}
Suppose by contradiction that there exists $L+1$ distinct polynomials $f_1,\dots,f_{L+1}\in\mathbb{F}_q[x]_{<k}$ and a received word $y\in(\mathbb{F}^s_q)^n$ such that for each $i\in[L+1]$, the codeword $c_i:=\mathcal{C}(f_i)$ of $f_i$ has relative hamming distance at most $\frac{L}{L+1}\left(1-\frac{sR}{s-L+1}\right)$ from $y$. Then for any $i\in[L+1]$, there is a subset $I_i\subseteq [n]$ where $|I_i|\ge \frac{n}{L+1}+\frac{Lk}{(L+1)(s-L+1)}$, such that for each $t\in I_i$, $c_{i}[t]=y[t]$. For each $j\in[n]$, we define $e_j=\left\{\vec{f}_i\colon j\in I_i\right\}$ as the set of polynomials whose codewords match with $y$ on position $j$ in this bad list. Let $(\mathcal{V},\mathcal{E})$ denote the corresponding geometric agreement hypergraph, there is $\mathcal{V}=\left\{\vec{f}_1,\dots,\vec{f}_{L+1}\right\}$, and $\mathcal{E}=(e_1,\dots,e_n)$.
The weight of this hypergraph can be bounded by
\begin{equation*}
\mathrm{wt}(\mathcal{V},\mathcal{E})\ge \sum_{j=1}^n\bigg(|e_j|-1\bigg)\ge\left(\sum^{L+1}_{i=1}|I_i|\right)-n\ge\frac{Lk}{s-L+1}.
\end{equation*}
Then, by \Cref{lem:non-distinctdesign}, $f_1,\dots,f_{L+1}$ cannot be distinct, which is a contradiction. We conclude that such a bad list $f_1,\dots,f_{L+1}$ doesn't exist.
\end{proof}
\end{document}